\documentclass[10pt,draftclsnofoot,onecolumn]{IEEEtran}
\usepackage{amsfonts}
\usepackage[dvips]{graphicx}
\usepackage{epsfig,latexsym,color}
\usepackage{float}
\usepackage{indentfirst}
\usepackage{amssymb}
\usepackage{amsmath}
\usepackage{times}
\usepackage{subfigure}
\usepackage{psfrag}
\usepackage{cite}
\usepackage{cases}

\newtheorem{lemma}{$\mathbf{Lemma}$}
\newtheorem{proposition}{$\mathbf{Proposition}$}

\begin{document}

\title{\textsc{Channel Estimation for Two-Way Relay Networks in the Presence of Synchronization Errors}
\author{\IEEEauthorblockN{Xinqian Xie, Mugen Peng, Yonghui Li, Wenbo Wang, and H. Vincent Poor}
\thanks{Copyright (c) 2014 IEEE. Personal use of this material is permitted. However, permission to use this material for any other purposes must be obtained from the IEEE by sending a request to pubs-permissions@ieee.org.}
\thanks{X. Xie, M. Peng, and W. Wang are with the Key Laboratory of
Universal Wireless Communications for Ministry of Education, Beijing
University of Posts \text{\&} Telecommunications, China (e-mail:
xxmbupt@gmail.com; pmg@bupt.edu.cn; wbwang@bupt.edu.cn). Y. Li is
with the School of Electrical and Information Engineering,
University of Sydney, Sydney, NSW, 2006, Australia (e-mail:
yonghui.li@sydney.edu.au). H. V. Poor is with the School of
Engineering and Applied Science, Princeton University, Princeton,
NJ, USA (e-mail: poor@princeton.edu).}
\thanks{The work of X. Xie, M. Peng and W. Wang was supported in part by the National Natural Science Foundation of China (Grant No. 61222103), the National Basic Research Program of China (973 Program) (Grant No. 2013CB336600), the State Major Science and Technology Special Projects (Grant No.
2013ZX03001001). The work of H. V. Poor was supported in part by the
National Science Foundation under Grant ECCS-1343210.}}}  \maketitle
\vspace{-0.6in}
\begin{abstract}
This paper investigates pilot-aided channel estimation for two-way
relay networks (TWRNs) in the presence of synchronization errors
between the two sources. The unpredictable synchronization error
leads to time domain offset and signal arriving order (SAO)
ambiguity when two signals sent from two sources are superimposed at
the relay. A two-step channel estimation algorithm is first
proposed, in which the linear minimum mean-square-error (LMMSE)
estimator is used to obtain initial channel estimates based on pilot
symbols and a linear minimum error probability (LMEP) estimator is
then developed to update these estimates. Optimal training sequences
and power allocation at the relay are designed to further improve
the performance for LMMSE based initial channel estimation. To
tackle the SAO ambiguity problem, the generalized likelihood ratio
testing (GLRT) method is applied and an upper bound on the SAO
detection error probability is derived. By using the SAO
information, a scaled LMEP estimation algorithm is proposed to
compensate the performance degradation caused by SAO detection
error. Simulation results show that the proposed estimation
algorithms can effectively mitigate the negative effects caused by
asynchronous transmissions in TWRNs, thus significantly
outperforming the existing channel estimation algorithms.
\end{abstract}
\vspace{-0.2in}
\newpage

\IEEEpeerreviewmaketitle
\section{INTRODUCTION}
Wireless relay technology has attracted considerable attention due to its capability of providing spatial diversity and extending system coverage\cite{01}. However, relay-assisted transmission consumes extra bandwidth, leading to a loss in system throughput\cite{02}. Network coding has been shown to be an effective technique to improve the spectral efficiency of wireless networks by allowing multiple nodes to transmit simultaneously\cite{03}. In two-way relay networks (TWRNs), physical-layer network coding (PNC) allows two sources to exchange information within two time-slots in contrast to the four time-slots transmission needed for the conventional relaying protocols, leading to a doubled system throughput\cite{04}.

However, most existing work on PNC in TWRNs has assumed the knowledge of perfect channel state information (CSI). The performance of PNC in TWRNs is actually very sensitive to channel estimation errors because instantaneous CSI is required not only for coherent detection but also for self-interference cancelation at each source\cite{05}. In most practical systems, the receiver acquires the CSI based on pilot symbols\cite{06}. Much work has been devoted to solving the channel estimation problem for relay channels using pilot-based solutions\cite{07}. In \cite{08}, \emph{Gao}, \emph{et al}. first studied the training based channel estimation issue for TWRNs and presented a novel and effective solution, which has laid the groundwork for the enhancement of channel estimation as well as the training design in network coded relaying channels. A general framework to estimate the unknown timing and channel parameters for cooperative relay systems was proposed in \cite{09}. In \cite{10}, a channel estimation method based on a complex exponential basis expansion model was developed for time-varying bidirectional relaying channels. All these works revealed that channel estimation for relay networks, especially TWRNs, is quite different from that for conventional point-to-point wireless systems\cite{11} \cite{12}.

Most existing channel estimation techniques for TWRNs have assumed perfect synchronization between the two source nodes, such that the two training sequences sent by the two sources can be perfectly superimposed at the relay\cite{13}. However, in practice, the two sources cannot synchronize with each other perfectly, and inevitable synchronization errors will occur\cite{14}. This causes misaligned reception of the two training sequences at the relay\cite{15}. Clearly, this unpredictable sequence misalignment may destroy the orthogonality of the two sequences and lead to severe estimation errors. In order to alleviate the effects of synchronization errors on channel estimation, the relay strategy and training sequences need to be redesigned. Furthermore, the sequence arriving order (SAO) ambiguity caused by symbol misalignment is another key issue to be solved in asynchronous transmissions. This requires the receiver to perform joint SAO detection and channel estimation because the source node needs to detect the SAO before estimating the channels.

To our knowledge, the impact of synchronization errors on channel estimation performance in TWRNs is still unclear, and how to compensate estimation performance degradation caused by unpredictable symbol misalignment and SAO ambiguity has not been addressed in the open literature. The motivation of this work is to develop effective channel estimation algorithms for typical half-duplex amplify-and-forward (AF) TWRNs in the presence of synchronization errors. A joint SAO detection and channel estimation algorithm is developed based on the generalized likelihood ratio testing (GLRT) \cite{16} by formulating this joint optimization problem as a composite hypothesis test\cite{17}. Our main contributions are summarized as follows:
\begin{itemize}
\item We develop a linear minimum error probability (LMEP) based estimation algorithm to minimize the error probability of data detection. We design the optimal training sequences and power allocation to further improve the estimation performance under synchronization errors.
\item By taking into account the SAO ambiguity, we formulate the channel estimation problem as a composite hypothesis test, which allows the GLRT method to be used for the SAO problem. A scaled LMEP algorithm is proposed to minimize the performance degradation caused by the SAO detection error.
\item The performance of the proposed algorithms is analyzed and evaluated. Results show that the proposed algorithms can effectively mitigate the negative effects caused by asynchronous transmissions, and achieve significant performance gain in contrast to the existing channel estimation algorithms.
\end{itemize}
The rest of the paper is organized as follows. Section II introduces the system model and transmission scheme. In Section III, an LMEP estimation algorithm is presented for the case of perfect SAO information. The SAO detection and a scaled LMEP algorithm are presented in Section IV. The simulation results are shown in Section V, followed by the conclusions in Section VI.

\textbf{Notations}: The transpose, inverse, pseudo-inverse and Hermitian of a matrix are denoted by $\left(\cdot\right)^{T}$, $\left(\cdot\right)^{-1}$, $\left(\cdot\right)^{\dag}$ and $\left(\cdot\right)^{H}$, respectively; $\|\cdot\|$ denotes the two-norm of a vector; $|\cdot|$, $\angle\left(\cdot\right)$ and $\mathfrak{R}\{\cdot\}$ denote the magnitude, phase and the real
part of a complex argument, respectively; $\lfloor\cdot\rfloor$ denotes the the largest integer that not greater than the real argument; $\mathrm{tr}\!\left(\cdot\right)$ denotes the trace of a matrix; $\otimes$ denotes the Kronecker product; $\mathbf{diag}\left(e_{1},e_{2},\ldots,e_{N}\right)$ is the $N\!\times\! N$ diagonal matrix. $\mathbf{I}_{K}$ represents the $K \! \times\! K$ dimensional unitary diagonal matrix. $\mathcal{E}\{\cdot\}$ denotes the expectation of random variables.

\section{System Model}
Consider bidirectional communication between two source nodes ($\mathbb{S}_{i},i\!=\!1,2$) with the assistance of a relay node $\mathbb{R}$ as shown in Fig.\ref{fig_0}. All the three nodes are equipped with a single antenna, and the relay operates in half-duplex mode. The channel coefficient between $\mathbb{S}_{i}$ and $\mathbb{R}$ is denoted by $h_{i}$, and the transmitting power of $\mathbb{S}_{i}$ is denoted by $P_{i}$. The symbol period is denoted by $T_{s}$, and $p\left(t - nT_{s}\right)$ denotes the pulse shaping function for the baseband signal as
\begin{subnumcases}
{p\left(t\right)=}
1, & $0\leq t \leq T_{s}$,\\
0, & else.
\end{subnumcases}
For easy reference, the key notations are summarized in Table I.
In a TWRN, $\mathbb{S}_{1}$ and $\mathbb{S}_{2}$ first simultaneously transmit signals to $\mathbb{R}$. The signal sent by $\mathbb{S}_{i}$ is composed of three parts:
\begin{itemize}
\item a pilot sequence containing $N$ symbols denoted by $\big\{t_{i}\!\left[n\right]\!\big\}_{n\!=\!1,\ldots,N}$ with $|t_{i}\!\left[n\right]|^{2} \!=\! 1$,
\item a guard space of $L$ symbol length, and
\item a data sequence of $M$ symbols denoted by $\big\{s_{i}\!\left[m\right]\!\big\}_{m\!=\!1,\ldots,M}$ with $\mathcal{E}\!\big\{|s_{i}\!\left[m\right]|^{2}\big\}\! = \!1$.
\end{itemize}
Owing to the imperfect synchronization, the two signal sequences sent from the two sources may not arrive at $\mathbb{R}$ at the same time.

We consider a quasi-static fading channel, for which the channels remain constant within one transmission block but vary from one block to another. The fading coefficient $h_{i}$ is assumed to be a circularly symmetric complex Gaussian random variable with zero mean and variance $\upsilon_{i}$, and $h_{1}$ and $h_{2}$ are independent from each other. Both the sources and relay are assumed to have the full knowledge of the training sequences $\mathbf{t}_{i}$, channel variance $\upsilon_{i}$ as well as the noise variance. Without loss of generality, to avoid interference between the data and pilot symbols, we assume that the signal transmitted by $\mathbb{S}_{1}$ arrives at the relay priori to that of $\mathbb{S}_{2}$ by a time offset ${\tau}$ where $0 \!\leq\! {\tau}\! \leq\! LT_{s}$.

\begin{table}
\caption{Key Notations} \center
\begin{tabular}{c p{0.4\textwidth} l}
\hline
$h_{i}$ & $\mathbb{S}_{i}$-$\mathbb{R}$ channel \\
$\upsilon_{i}$ & Variance of $h_{i}$\\
$h_{a}$ & $\triangleq h_{1}^{2}$ \\
$\upsilon_{a}$ & Variance of $h_{a}$\\
$h_{b}$ & $\triangleq h_{1}h_{2}$ \\
$\upsilon_{b}$ & Variance of $h_{b}$\\
$\mathbf{h}$ & $\triangleq \left[h_{1}, h_{2}\right]^{T}$ \\
$\mathbf{\Theta}$ & $\triangleq\left[h_{a}, h_{b}\right]^{T}$\\
$\delta_{\Sigma}$ & MSE of $\hat{\mathbf{\Theta}}$\\
$T_{s}$ & Symbol period\\
$N_{S0}$ & Noise power at $\mathbb{S}_{i}$\\
$N_{R0}$ & Noise power at $\mathbb{R}$\\
$\mathbf{t}_{i}$ & Training sequence from $\mathbb{S}_{i}$ \\
$N$ & Training sequence length \\
$\tau$ & Time domain offset\\
$n_{\tau}$ & $\triangleq\lfloor\frac{\tau}{T_{s}}\rfloor$\\
$\lambda$ & $\triangleq\tau\!-\!n_{\tau}T_{s}$\\
$\vartheta$ & Sequence arriving order\\
$\mathcal{H}_{\vartheta}$ & Hypothesis $\vartheta$\\
\hline
\end{tabular}
\end{table}

\subsection{Pilot Transmission}
The received pilot signal at $\mathbb{R}$ can be expressed as
\begin{align}
x_{\!R}\!\left(t\right) \!=\!
\sum_{n=1}^{N}\!\bigg\{\!\!\sqrt{\!P_{\!1}}h_{\!1}t_{\!1}\!\left[n\right]\!p\big(t \!-\!
\left(n\! \!-\! \!1\right)\!T_{s}\big)\! +\!
\sqrt{\!P_{2}}h_{2}t_{2}\!\left[n\right]\!p\big(t \!-\! \left(n \!\!-\!\!
1\right)\!T_{s}\! -\! \tau\big)\!\bigg\}\! + \!w_{R}\!\left(t\right),
\end{align}
where $w_{R}\!\left(t\right)$ represents additive white Gaussian noise (AWGN) at $\mathbb{R}$ with the average power $N_{R0}$. Before forwarding the signals to the two sources, the relay amplifies $x_{R}\!\left(t\right)$ by a function $\gamma\!\left(t\right)$
\begin{subnumcases}
{\gamma\left(t\right)\!=\!}
\gamma_{1}, & $0\!\leq\! t \!\leq\! \tau$,\\
\gamma_{S}, & $\tau \!<\! t\!<\! N T_{s}$,\\
\gamma_{2},& $N T_{s}\!\leq\! t \!\leq\! N T_{s}\! + \!\tau$.
\end{subnumcases}
Thus, the signals forwarded by the relay satisfy the following energy constraint at relay
\begin{align}\label{PowerConstraint}
\sum_{i=1}^{2}\!\gamma_{i}^{2}\frac{\tau}{T_{s}}\big(\upsilon_{i}P_{i}T_{s}\!+\!N_{R0}\big)
\!+\!\gamma_{S}^{2}\left(N\!-\!\frac{\tau}{T_{s}}\right)\!\left[\sum_{i=1}^{2}\upsilon_{i}
P_{i}T_{s}\!+\!N_{R0}\right]\!=\!E_{r},
\end{align}
where $E_{r}$ represents the energy consumption of the pilot signals at $\mathbb{R}$ per block. Without loss of generality, let us consider the problem at $\mathbb{S}_{1}$, and that of $\mathbb{S}_{2}$ can be handled in similar way. The received pilot signal at $\mathbb{S}_{1}$ is given by
\begin{align}
x_{S\!1}\!\left(t\right)\!=\!\gamma\!\left(t\right)\!
\sum_{n=1}^{N}\!\bigg\{\!\!\sqrt{\!P_{1}}h_{1}^{2}t_{1}\!\left[n\right]p\big(t\!
- \!\left(n \!\!-\!\! 1\right)\! T_{s}\big) \!+\!
\sqrt{\!P_{2}}h_{1}h_{2}t_{2}\!\left[n\right]p\big(t\! -\! \left(n\!\! -\!\!
1\right)\!T_{s} \!-\! \tau\big)\!\!\bigg\} \!+\!
\gamma\!\left(t\right)\!w_{R}\!\left(t\right)\! +\! w_{S}\!\left(t\right)
\end{align}
where $w_{S1}\!\left(t\right)$ is AWGN at $\mathbb{S}_{1}$ with the average power $N_{S0}$.

\subsection{Data Transmission}
Similarly, the received signal at $\mathbb{S}_{1}$ can be written as
\begin{align}
y_{S}\!\left(t\right)\!=\!\alpha\!
\sum_{m=1}^{M}\!\bigg\{\!\!\underbrace{\sqrt{\!P_{1}}h_{1}^{2}s_{1}\!\left[m\right]p\big(t
\!-\! \left(m\!\! -\!\! 1\right)\! T_{s}\big)}_{\text{self-interference}}\! +\!
\sqrt{\!P_{2}}h_{1}h_{2}\underbrace{s_{2}\!\left[m\right]p\big(t\! -\!
\left(m \!\!-\!\! 1\right)\!T_{s}\! -\! \tau\big)}_{\text{desired-signal}}\!\bigg\} \!+\! \alpha n_{R}\!\left(t\right)\! +\! n_{S}\!\left(t\right)
\end{align}
where $n_{R}\!\left(t\right)$ represents AWGN with variance $N_{R0}$ at $\mathbb{R}$, and $n_{S}\!\left(t\right)$ is AWGN with variance $N_{S0}$ at $\mathbb{S}_{1}$. The power scaling factor $\alpha$ is given by
\begin{align}
\alpha \!=\! \sqrt{\frac{P_{r}}{\upsilon_{1} P_{1} \!+\! \upsilon_{2} P_{2}
\!+\! N_{R0}}},
\end{align}
which ensures that the transmit power of data signals does not exceeds the power constraint $\left(MT_{s}\! +\! \tau\right)P_{r}$. Since $\mathbb{S}_{1}$ has the knowledge of $s_{1}\!\left[m\right]$, in order to remove the self-interference and recover $s_{2}\!\left[m\right]$, $\mathbb{S}_{1}$ needs to estimate the instantaneous composite channel coefficients $h_{a} \!\triangleq\! h_{1}^{2}$ and $h_{b} \!\triangleq \!h_{1} h_{2}$. Meanwhile, SAO ambiguity may occur when the two signals are superimposed at $\mathbb{R}$. The source node needs to estimate the required channel coefficients and detect the SAO.

\begin{figure}[!h]
\center
  \includegraphics[width=0.8\textwidth]{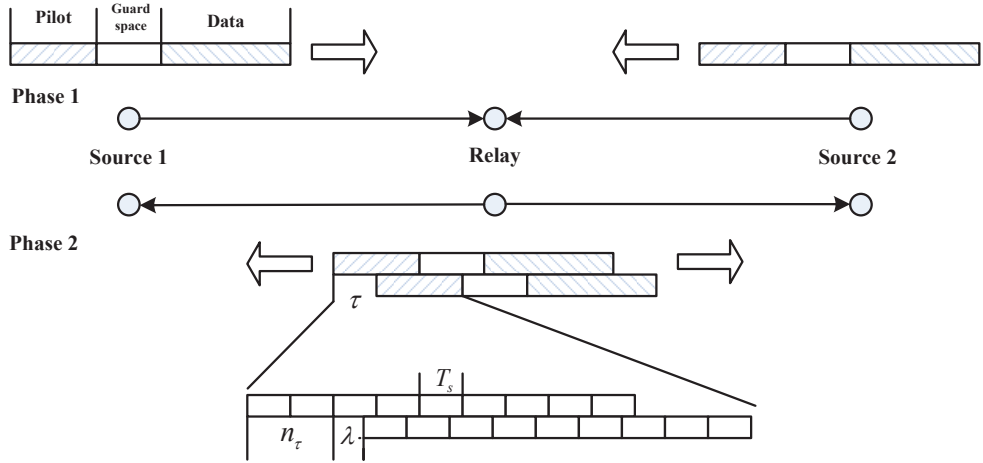}
 \caption{System model and transmission scheme.}\label{fig_0}
\end{figure}

\section{Channel Estimation in the Presence of Synchronization Errors}
In this section, we develop the channel estimation algorithm for the TWRN with synchronization errors with given SAO. Let $\tau$ represent the time domain misalignment duration between the two signals and it can be expressed as
\begin{align}\label{newEQ2}
\tau \!=\! \underbrace{\lfloor\tau/T_{s}\rfloor}_{n_{\tau}} \cdot T_{s}
+ \underbrace{\tau\! -\! n_{\tau} T_{s}}_{\lambda}.
\end{align}
For simplicity, we assume $P_{s}\! =\! P_{1}\! =\! P_{2}$ and $\upsilon\! =\! \upsilon_{1}\! =\! \upsilon_{2}$ leading to $\gamma_{1}\! =\!\gamma_{2}\! = \!\gamma_{I}$. The signal samples, denoted by $x_{S\!1}\!\left[k\right]_{k\!=\!1,\ldots,2N+1}$, can be expressed as follows \cite{15}:

For $1\!\leq\! k\! \leq\! n_{\tau}$,
\begin{align}
x_{S\!1}\!\left[k\right]
\!=\!\gamma_{I}h_{a}t_{1}\!\left[k\right]\!
+ \!\gamma_{I}h_{1}w_{R}\!\left[k\right] \!+\! w_{S\!1}\!\left[k\right].
\end{align}

For $k\!=\!n_{\tau}\!+\!2i\!-\!1$ with $1\!\leq \!i\! \leq\! \left(N\!-\!n_{\tau}\!+\!1\right)$,
\begin{align}
x_{S\!1}\!\left[k\right]\!=\!\gamma_{S}\sqrt{\frac{\lambda}{T_{s}}}\big(h_{a}t_{1}\!\left[i\!+\!n_{\tau}\right]
\!+\! h_{b}t_{2}\!\left[i\! -\! 1\right]\big)\! +\! \gamma_{S}h_{1}w_{R}\!\left[k\right] \!+\! w_{S1}\!\left[k\right],
\end{align}
where $t_{1}\!\left[N\!+\!1\right] \!=\! 0$ and $t_{2}\!\left[0\right] \!=\! 0$.

For $k\!=\!n_{\tau}\!+\!2i$ with $1 \!\leq \!i\! \leq \!\left(N\!-\!n_{\tau}\right)$,
\begin{align}
x_{S\!1}\!\left[k\right] \!=\!
\gamma_{S}\sqrt{\frac{T_{s}\!-\!\lambda}{T_{s}}}\big(h_{a}t_{1}\!\left[i\!+\!n_{\tau}\right]
\!+\! h_{b}t_{2}\!\left[i\right]\big)\! +\!\gamma_{S}h_{1}w_{R}\!\left[k\right] \!+\! w_{S\!1}\!\left[k\right].
\end{align}

For $\left(2N\!-\!n_{\tau}\!+\!2\right)\! \leq \!k\! \leq \!2N \!+\! 1$,
\begin{align}
x_{S\!1}\!\left[k\right]\!=\!
\gamma_{I}h_{b}t_{2}\left[k\!-\!N\!-\!1\right]\!
+\! \gamma_{I}h_{1}w_{R}\!\left[k\right]\! +\! w_{S\!1}\!\left[k\right].
\end{align}

Hence, $w_{R}\!\left[k\right]$ and $w_{S\!1}\!\left[k\right]$ satisfy the complex Gaussian distribution with mean zero and variance $\frac{N_{R0}}{P_{s}T_{s}}$ and $\frac{N_{S0}}{P_{s}T_{s}}$, respectively. Let us denote $\mathbf{t}_{i}$ to be the vector by stacking $t_{i}\!\left[n\right]$'s. Note that, only the first $\left(2N\!-\!\tau\!+\!1\right)$ observations of $\mathbf{x}_{S\!1}\!\left[k\right]$'s contain the entries of $\mathbf{t}_{1}$, while the entries of $\mathbf{t}_{2}$ are only in the last $\left(2N\!-\!\tau\!+\!1\right)$ observations of $\mathbf{x}_{S\!1}\!\left[k\right]$'s. In order to stack $x_{S\!1}\!\left[k\right]_{k\!=\!1,\ldots,2N+1}$ in vector form, let us first define the equivalent training sequences of length $\left(2N\!+\!1\right)$ as
\begin{align}
\mathbf{r}_{1}\! &=\!
\left[\mathbf{t}_{1}\!\left(1\!:\!n_{\tau}\right)^{T},\mathbf{t}_{1}\!\left(n_{\tau}\!+\!1\!:\!N\right)^{T}
\!\!\otimes\!\mathbf{J},\mathbf{0}_{n_{\tau}}\right]^{T},\\
\mathbf{r}_{2}\! &=\!
\left[\mathbf{0}_{n_{\tau}},\mathbf{t}_{2}\!\left(1\!:\!N\!
-\!n_{\tau}\right)^{T}\!\!\otimes\!\mathbf{J},\mathbf{t}_{2}\!\left(N\!
-\!n_{\tau}\!+\!1\!:\!N\right)^{T}\right]^{T},
\end{align}
where $\mathbf{J}\!\triangleq\!\left[1,1\right]$,
$\mathbf{t}_{i}\!\left(n_{1}\!:\!n_{2}\right)\!\triangleq\!\big[t_{i}[n_{1}],t_{i}[n_{1}\!+\!1],\ldots,
t_{i}[n_{2}]\big]^{T}$ represents the vector of length $\left(n_{2}\!-\!n_{1}\!+\!1\right)$ that contains the $n_{1}$th to $n_{2}$th entries of $\mathbf{t}_{i}$, and $\mathbf{0}_{n_{\tau}}$ denotes the zero vector of $1 \!\times \!n_{\tau}$ dimension. Then we can write
\begin{align}
\mathbf{x}_{S\!1}\!=\!\mathbf{\Gamma}\mathbf{\Lambda}\left[\mathbf{r}_{1},
\mathbf{r}_{2}\right]\underbrace{\left[h_{a},h_{b}\right]^{T}}_{\mathbf{\Theta}}\!+\!
h_{1}\mathbf{\Gamma}\mathbf{w}_{R}\!+\!\mathbf{w}_{S\!1},
\end{align}
where $\mathbf{\Gamma}$ and $\mathbf{\Lambda}$ are both $\left(2N\!+\!1\right)\!\times\!\left(2N\!+\!1\right)$ dimensional diagonal matrices,
\begin{align}
\mathbf{\Gamma}\! =\!
\mathbf{diag}\!\left[\gamma_{I}\mathbf{I}_{n_{\tau}},\gamma_{S}\mathbf{I}_{2\left(N
\!-\! n_{\tau}\right)\!+\!1},\gamma_{I}\mathbf{I}_{n_{\tau}}\right],
\end{align}
\begin{align}
\mathbf{\Lambda}\!=\!
\mathbf{diag}\!\!\left[\mathbf{I}_{n_{\tau}},\mathbf{I}_{2\!\left(N\!-\!n_{\!\tau}\right)}\!\otimes\!\mathbf{diag}
\!\left(\!\!\sqrt{\!\frac{\lambda}{T_{s}}},
\sqrt{\frac{T_{s}\!-\!\lambda}{T_{s}}}\right),
\sqrt{\!\frac{\lambda}{T_{s}}},\mathbf{I}_{n_{\tau}}\!\right],
\end{align}
while $\mathbf{w}_{R}$ and $\mathbf{w}_{S\!1}$ are $\left(2N\!+\!1\right)\!\times\!1$ dimensional noise vectors by stacking $w_{R}\!\left[k\right]$'s and $w_{S\!1}\!\left[k\right]$'s, respectively.

\subsection{Linear Minimum Error Probability Estimation Algorithm}
In this subsection, we present a two-step estimation algorithm to minimize the error probability for coherent reception, where the LMMSE estimator is used to obtain the initial estimates of $h_{a}$ and $h_{b}$, after which the estimates are updated by using the LMEP estimator.

By definition, the LMMSE estimator, denoted by $\mathbf{D}$, can be easily calculated as
\begin{align}
\mathbf{D}\! =\! \mathbf{diag}\!\left(\upsilon_{a}, \upsilon_{b}\right)\!\left[\mathbf{r}_{1},
\mathbf{r}_{2}\right]^{H}\!\!\mathbf{\Lambda}\mathbf{\Gamma}\bar{\mathbf{R}}_{S}^{-1},
\end{align}
where $\upsilon_{a}$ and $\upsilon_{b}$ are the variances of $h_{a}$ and $h_{b}$, respectively. The $\left(2N\!+\!1\right)\!\times\!\left(2N\!+\!1\right)$ dimensional matrix $\bar{\mathbf{R}}_{S}$ is the covariance matrix of $\mathbf{x}_{S\!1}$, which is given by
\begin{align}
\bar{\mathbf{R}}_{S}\! =\!
\upsilon_{a}\mathbf{\Gamma}\mathbf{\Lambda}\mathbf{r}_{1}\mathbf{r}_{1}^{H}\mathbf{\Lambda}\mathbf{\Gamma}
\!+\!
\upsilon_{b}{\mathbf{\Gamma}}\mathbf{\Lambda}\mathbf{r}_{2}\mathbf{r}_{2}^{H}\mathbf{\Lambda}\mathbf{\Gamma}
\!+\! \left(\!\frac{\upsilon_{1}N_{R0}}{P_{s}T_{s}}\mathbf{\Gamma}^{2}
\!+\!\frac{N_{S0}}{P_{s}T_{s}}\mathbf{I}_{2N\!+\!1}\right).
\end{align}
Thus the initial estimates of $h_{a}$ and $h_{b}$ can be obtained by
\begin{align}\label{EQ11}
\left[\hat{h}_{a},\hat{h}_{b}\right]^{T}\! \!=\! \mathbf{D}\mathbf{x}_{S\!1}.
\end{align}

Since the instantaneous bit error probability (BEP) for coherent reception can be approximately computed by $\mathcal{Q}\big(\!\!\sqrt{\beta\!\cdot\!\Upsilon}\big)$, where $\beta$ is relevant to the specific modulation and $\Upsilon$ represents the received effective signal-to-noise ratio (SNR) at $\mathbb{S}_{1}$ given by
\begin{align}\label{effective SNR}
\mathrm{\Upsilon}\! =\!
\frac{\mathcal{E}\Big\{|\hat{h}_{b}|^{2}\Big\}}{\mathcal{E}\!\Big\{|\hat{h}_{a}\! - \!h_{a}|^{2}\!+\!
|\hat{h}_{b}\! -\! h_{b}|^{2}\Big\}\! +\!
\left(\!\frac{|h_{1}|^{2}N_{\!R0}}{P_{s}T_{s}} \!+\!
\frac{N_{\!S0}}{\alpha^{2}P_{s}T_{s}}\right)},
\end{align}
where the expectation is taken only with respect to noise terms. Clearly, $\mathcal{Q}\big(\!\!\sqrt{\beta\!\cdot\!\Upsilon}\big)$ decreases with respect to $\Upsilon$. Thus minimizing BEP is equivalent to maximizing $\Upsilon$.

Let $\mathbf{u}_{a}$ and $\mathbf{u}_{b}$ denote the $\left(2N\!+\!1\right)\!\times\!1$ dimensional receiver processing vector to estimate the channel coefficients from the received signals
$\hat{h}_{a}\!=\!\mathbf{u}^{H}_{a}\mathbf{x}_{S}$ and $\hat{h}_{b} \!=\!\mathbf{u}^{H}_{b}\mathbf{x}_{S}$, respectively. It can be observed from \eqref{effective SNR} that only the term
$\epsilon_{a}\!\left(\mathbf{u}_{a}\right)\!=\!\mathcal{E}\!\big\{|\mathbf{u}^{H}_{a}\mathbf{x}_{S1}
\!-\! h_{a}|^{2}\big\}$ depends on $\mathbf{u}_{a}$; so we can obtain the optimal $\mathbf{u}_{a}$, denoted by $\mathbf{u}^{*}_{a}$, by minimizing $\epsilon_{a}\left(\mathbf{u}_{a}\right)$. It can be easily calculated that
\begin{align}\label{optimal_u_a}
\mathbf{u}_{a}^{*}\! =\!
\mathbf{R}_{S}^{-1}\mathbf{\Gamma}\mathbf{\Lambda}\left(|h_{a}|^{2}\mathbf{r}_{1}
\!+\! h_{a}^{H}h_{b}\mathbf{r}_{2}\right),
\end{align}
where
\begin{align}
\mathbf{R}_{S} \!=\!
\mathbf{\Gamma}\mathbf{\Lambda}\left(h_{a}\mathbf{r}_{1}\!+\!h_{b}\mathbf{r}_{2}\right)\!
\left(h_{a}\mathbf{r}_{1}\!+\!h_{b}\mathbf{r}_{2}\right)^{H}\mathbf{\Lambda}\mathbf{\Gamma}\!+\! \left(\frac{|h_{1}|^{2}N_{R0}}{P_{s}T_{s}}\mathbf{\Gamma}^{2}
\!+\!\frac{N_{S0}}{P_{s}T_{s}}\mathbf{I}_{2N\!+\!1}\right).
\end{align}
By substituting \eqref{optimal_u_a} into $\epsilon_{a}$, we have
\begin{align}
\underbrace{\min_{\mathbf{u}_{a}}{\epsilon_{a}}}_{\epsilon_{a}^{*}}
\!=\! |h_{a}|^{2}\!-\!\left(|h_{a}|^{2}\mathbf{r}_{1}\! +\!
h_{a}^{H}h_{b}\mathbf{r}_{2}\right)^{H}\!\mathbf{\Lambda}\mathbf{\Gamma}
\mathbf{R}_{S}^{-1}\mathbf{\Gamma}\mathbf{\Lambda}\left(|h_{a}|^{2}\mathbf{r}_{1}
\!+\! h_{a}^{H}h_{b}\mathbf{r}_{2}\right).
\end{align}
Substituting $\epsilon_{a}^{*}$ into \eqref{effective SNR} and after some manipulation, $\Upsilon$ can be written as
\begin{align}\label{update_effecive_SNR}
\Upsilon\!\left(\mathbf{u}_{b}\right)\! =\!
\frac{\mathbf{u}_{b}^{H}\mathbf{R}_{S}\mathbf{u}_{b}}
{\mathbf{u}^{H}_{b}\mathbf{R}_{S}\mathbf{u}_{b}\! -\!
\mathbf{u}_{b}^{H}\mathbf{r}_{E} \!-\! \mathbf{r}_{E}^{H}\mathbf{u}_{b}
\!+\! A},
\end{align}
where
\begin{align}
\mathbf{r}_{E}\! =\!
\mathbf{\Gamma}\mathbf{\Lambda}\left(h_{a}h_{b}^{H}\mathbf{r}_{1} \!+\!
|h_{b}|^{2}\mathbf{r}_{2}\right),
\end{align}
\begin{align}
A\! =\! \epsilon_{a}^{*} \!+\! |h_{b}|^{2} \!+\! \left(\frac{|h_{1}|^{2}N_{R0}}{P_{s}T_{s}} \!+\!
\frac{N_{S0}}{\alpha^{2}P_{s}T_{s}}\right).
\end{align}
Hence, the optimal $\mathbf{u}_{b}$ can be obtained by maximizing $\Upsilon\!\left(\mathbf{u}_{b}\right)$. On setting $A_{1} \!=\! \mathbf{u}_{b}^{H}\mathbf{R}_{S}\mathbf{u}_{b} $
and $A_{2} \!=\! \mathbf{u}_{b}^{H}\mathbf{r}_{E}\! +\! \mathbf{r}_{E}^{H}\mathbf{u}_{b} \!=\!
2\mathfrak{R}\{\mathbf{r}_{E}^{H}\mathbf{u}_{b}\}$, we can further denote $\mathbf{r}_{E}^{H}\mathbf{u}_{b} \!=\! \frac{A_{2}}{2}\! +\! j A_{3}$ with a real $A_{3}$ and arrive at
\begin{align}
\mathbf{u}_{b} \!=\! \left(\frac{A_{2}}{2}\! +\! j
A_{3}\right)\frac{\mathbf{r}_{E}}{\|\mathbf{r}_{E}\|^{2}}.
\end{align}
Substituting the above expression for $\mathbf{u}_{b}$ into $A_{1}$ and after some mathematical manipulation, we have
\begin{align}
A_{2} \!=\!
2\sqrt{\mathbf{r}_{E}^{H}\mathbf{R}_{S}^{-1}\mathbf{r}_{E}\left(A_{1}
\!-\! A_{3}^{2}\right)}.
\end{align}
Since $\Upsilon\!\left(\mathbf{u}_{b}\right)$ is an increasing function of $A_{2}$, the optimal $\mathbf{u}_{b}$ is obtained when $A_{2}$ reaches its maximum. To do so, $A_{3}$ needs to be set as zero. In this case, the optimization problem simplifies to
\begin{align}
\max_{A_{1}}\frac{A_{1}}{A_{1} \!-\!
2\sqrt{\mathbf{r}_{E}^{H}\mathbf{R}_{S}^{-1}\mathbf{r}_{E}A_{1}} \!+\!
A}.
\end{align}
It can be easily checked that the optimal $A_{1}$ satisfies $\sqrt{\!A^{*}_{1}}\! =\!
\frac{A}{\sqrt{\mathbf{t}_{E}^{H}\mathbf{R}_{x2}^{-1}\mathbf{t}_{E}}}$, and thus we have $A_{2}\!=\!2A$. As a result, the optimal $\mathbf{u}^{*}_{b}$ can be further written as
\begin{align}\label{EQ1}
\mathbf{u}^{*}_{b} \!=\! \frac{A\mathbf{r}_{E}}{\|\mathbf{r}_{E}\|^{2}}.
\end{align}
Since both $\mathbf{u}^{*}_{a}$ and $\mathbf{u}^{*}_{b}$ depend on the instantaneous channel coefficients, we can use the initial estimates $\hat{h}_{a}$ and $\hat{h}_{b}$ in place of $h_{a}$ and $h_{b}$ to construct the LMEP estimators.

The LMEP estimation algorithm is summarized in Table II.

\begin{table}
\caption{LMEP Estimation Algorithm} \center
\begin{tabular}{p{5in}@{}}
\hline
\begin{itemize}
\item \textbf{Step 1}
\begin{itemize}
\item[-] Calculate the LMMSE estimator $\mathbf{D}$.
\item[-] Obtain the initial estimates of $\hat{h}_{a}$ and $\hat{h}_{b}$ as $\left[\hat{h}_{a},\hat{h}_{b}\right]^{T}\! =\! \mathbf{D}\mathbf{x}_{S1}$.
\end{itemize}
\item \textbf{Step 2}
\begin{itemize}
\item[-] Construct $\mathbf{u}^{*}_{a}$ and $\mathbf{u}^{*}_{b}$ by using $\hat{h}_{a}$ and $\hat{h}_{b}$.
\item[-] Update the estimate of $h_{a}$ and $h_{b}$ as $\hat{h}_{a}\! =\! \left(\mathbf{u}_{a}^{*}\right)^{H}\mathbf{x}_{S1}$ and $\hat{h}_{b} \!=\! \left(\mathbf{u}_{b}^{*}\right)^{H}\mathbf{x}_{S1}$.
\end{itemize}
\item \textbf{Return} $\left[\hat{h}_{a},\hat{h}_{b}\right]^{T}$.
\end{itemize}
\\
\hline
\end{tabular}
\end{table}

\subsection{Design of Training Sequence and Power Scaling at Relay}
We can note from the above description that the performance of the LMEP estimation algorithm depends critically on the accuracy of the initial estimates of $h_{a}$ and $h_{b}$; thus in this subsection we further improve the performance of initial LMMSE estimation by designing optimal training sequences and power allocation. By definition, the MSE summation of $\hat{h}_{a}$ and $\hat{h}_{b}$ is given by
\begin{align}
\delta_{\Sigma}\! &=\!
\mathrm{tr}\bigg\{\mathcal{E}\!\left[\left(\mathbf{D}\mathbf{x}_{S1}\! -\!
\mathbf{\Theta}\right)\!\left(\mathbf{D}\mathbf{x}_{S1}\! -\!
\mathbf{\Theta}\right)^{H}\right]\!\bigg\}\\
\!&=\! \upsilon_{a}\!+\!\upsilon_{b}\!-\!
\upsilon_{a}^{2}\mathbf{r}_{1}^{H}\mathbf{\Lambda}\mathbf{\Gamma}
\bar{\mathbf{R}}_{S}^{-1}\mathbf{\Gamma}\mathbf{\Lambda}\mathbf{r}_{1}\!-\!
\upsilon_{b}^{2}\mathbf{r}_{2}^{H}\mathbf{\Lambda}\mathbf{\Gamma}
\bar{\mathbf{R}}_{S}^{-1}\mathbf{\Gamma}\mathbf{\Lambda}\mathbf{r}_{2}.\nonumber
\end{align}
For simplicity, we further assume $N_{0}\!=\!N_{R0}\!=\!N_{S0}$ and the derivation for the general case is straightforward. By applying the Woodbury's identity and after some mathematical calculations, $\delta_{\Sigma}$ can be written as
\begin{align}\label{MSE1}
{\delta}_{\Sigma}\! =\! \frac{\upsilon_{a} \!+\! \upsilon_{b} \!+\!
2\upsilon_{p}B_{1}}{1 \!+\! \underbrace{\left(\upsilon_{a}\! +\!
\upsilon_{b}\right)}_{\upsilon_{s}}B_{1} \!+\! \underbrace{\upsilon_{a}
\upsilon_{b}}_{\upsilon_{p}}B_{1}^{2} \!-\! \upsilon_{a}
\upsilon_{b}|B_{2}|^{2}},
\end{align}
where
\begin{align}
B_{1} \!&= \!\frac{\left(NT_{s} \!-\!
{\tau}\right)P_{s}\gamma_{S}^{2}}{N_{0}\!\left(1 \!+\!
\upsilon\gamma_{S}^{2}\right)}\! + \!\frac{
{\tau}P_{s}\gamma_{I}^{2}}{N_{0}\!\left(1 \!+\!
\upsilon\gamma_{I}^{2}\right)}, \\
|B_{2}|\! &=\!
\frac{|\rho\left({\tau}\right)|NT_{s}P_{s}\gamma_{S}^{2}}{N_{0}\!\left(1
\!+\! \upsilon\gamma_{S}^{2}\right)},
\end{align}
and $\rho\left(\tau\right)\! =\!
\frac{\mathbf{r}_{1}^{H}\mathbf{\Lambda}^{2}\mathbf{r}_{2}}{\|\mathbf{\Lambda}\mathbf{r}_{1}\|\cdot\|\mathbf{\Lambda}\mathbf{r}_{2}\|}$.

\subsubsection{Training sequence design}
It can be noted that $\delta_{\Sigma}$ increases with respect to $|\rho\left(\tau\right)|$, thus the training sequences minimizing the resulting MSE should satisfy $\rho\left({\tau}\right)\! =\!
0$ for arbitrary ${\tau}$, which is not always achievable. However, when $N$ is sufficiently large, we can have the following proposition.
\begin{proposition}
For a sufficiently large $N$, the training sequences of the two sources in the TWRN that minimize the corresponding estimation MSE can be selected from any two different columns of the $N\times N$ discrete Fourier transform (DFT) matrix, since this choice yields
\begin{align}
\lim_{N\rightarrow \infty}{|\rho\left(\tau\right)|} \!=\! 0.
\end{align}
\end{proposition}
\begin{proof}
See Appendix A.
\end{proof}
While for a finite $N$, we focus on obtaining the training sequences that minimizes the maximum $|\rho\left({\tau}\right)|$ for $\tau \!\in\! \left[0, NT_{s}\right]$, which can be formulated as
\begin{align}
\min_{\mathbf{t}_{1}, \mathbf{t}_{2}} \max_{{\tau}\in\left[0,
NT_{s}\right]} |\rho\left({\tau}\right)|.
\end{align}
Note that, each entry of $\mathbf{t}_{i}$ satisfies $|t_{i}\!\left[n\right]|\!=\!1$ and has a arbitrary phase, $\angle{t_{i}\!\left[n\right]}\!\in\!\left[0, 2\pi\right)$. It is effective to solve such a problem through a $2N \!\times\! 2N$ dimensional search, but this will introduce a prohibitively high complexity. Since limiting the training sequence to be one column of the $N \!\times\! N$ DFT matrix has been proved to be optimal design for sufficiently large $N$, in this paper we focus on the design of training sequence by limiting them to the columns of $N\!\times\!N$ DFT matrices, so that we can achieve a better tradeoff between estimation performance and computational complexity. Since there are $\frac{N\left(N \!-\! 1\right)}{2}$ different choices for selecting the training sequences from the columns of the $N\times N$ DFT matrix, and the performance with different selections differs significantly, we have the following proposition about the optimal selection.
\begin{proposition}
Let $\mathbf{t}_{i}$ denote the $k_{i}$-th column of ${N \!\times\! N}$ DFT matrix, selected as the training sequence of $\mathbb{S}_{i}$ where $k_{i}\!\in\!\left[1, N\right]$ and $k_{1}\!<\!k_{2}$,
\begin{itemize}
\item when $N$ is even, the optimal training sequences should satisfy $k_{2} \!-\! k_{1} \!=\! \frac{N}{2}$,
\item when $N$ is odd, the optimal training sequences should satisfy $k_{2} \!-\! k_{1} \!=\! \frac{N\pm1}{2}$.
\end{itemize}
\end{proposition}
\begin{proof}
See in appendix B.
\end{proof}

\subsubsection{Power allocation}
By using the optimal training sequences specified above, we have $\rho\left({\tau}\right) \!\leq\! \frac{1}{N}$ for arbitrary $\tau$ leading to $|B_{2}| \!\leq\! \frac{P_{s}T_{s}\gamma_{S}^{2}}{N_{0}\left(1 + \upsilon \gamma_{S}^{2}\right)}$. In practice, time domain offset is much less than the length of pilot signal, i.e., ${\tau}\!\ll\! NT_{s} $, thus we have $|B_{2}| \!\ll\! B_{1}$. Hence, minimizing $\delta_{\Sigma}$ can be approximately expressed as
\begin{align}\label{EQ8}
\min_{\gamma_{I}^{2}, \gamma_{S}^{2}} f\!\left(B_{1}\right)\!=\!\frac{\upsilon_{s} \!+\!
2\upsilon_{p}B_{1}}{1 \!+\! \upsilon_{s}B_{1} \!+\! \upsilon_{p}B_{1}^{2}}.
\end{align}
It can be testified that the objective function $f\!\left(B_{1}\right)$ satisfies $\frac{\partial f\!\left(B_{1}\right)}{\partial B_{1}}\!<\!0$ for $B_{1}\!>\!0$, which indicates that $f\!\left(B_{1}\right)$ is decreasing with respect to $B_{1}$, thus \eqref{EQ8} is equivalent to maximizing $B_{1}$ as
\begin{align}
\max_{\gamma_{I}^{2}, \gamma_{S}^{2}}& \quad B_{1}\!=\!\frac{\left(NT_{s}\! -\!
{\tau}\right)P_{s}\gamma_{S}^{2}}{N_{0}\left(1 \!+\!
\upsilon\gamma_{S}^{2}\right)}\!+\!\frac{
{\tau}P_{s}\gamma_{I}^{2}}{N_{0}\left(1 \!+\!
\upsilon\gamma_{I}^{2}\right)}\\
\mathrm{s.t.}& \quad\eqref{PowerConstraint}
\end{align}
The above optimization problem is concave and it can be solved by using Lagrange multiplier method, and the scaling factors are obtained as
\begin{align}
\gamma_{I}^{2} \!&=\! \frac{\left(NT_{s} \!\!-\!\! {\tau}\right)\!E_{b}\! +\!
\upsilon T_{s}E_{r}\! \!-\!\! \sqrt{2E_{a}E_{b}}\!\left(NT_{s}\!\! -\!\!
{\tau}\right)}{\sqrt{2E_{a}E_{b}}\left(NT_{s} \!\!-\!\!
{\tau}\right)\upsilon \!+\! 2{\tau}E_{a}\upsilon},\\
\gamma_{S}^{2}\!&=\! \frac{2{\tau}E_{a}\! +\! \upsilon E_{r}T_{s} \!\!-\!\!
\sqrt{2E_{a}E_{b}}{\tau}}{\sqrt{2E_{a}E_{b}} {\tau}\upsilon \!+\!
\left(NT_{s} \!\!-\!\! {\tau}\right)E_{b}\upsilon},
\end{align}
where $E_{a}\! =\! \upsilon P_{s}T_{s} \!+\! N_{0}$ and $E_{b}\! = \!2\upsilon P_{s} T_{s}\!+\! N_{0}$. This is the sub-optimal power allocation in minimizing estimation MSE.

By substituting above scaling factors into $B_{1}$, it yields
\begin{align}\label{newEQ1}
B_{1}\! =\! \frac{NT_{s}E_{r} \!+\! \tau\left(NT_{s}\! \!-\!\!\tau\right)\left(\sqrt{2E_{a}}\!\! -\!\! \sqrt{E_{b}}\right)^{2}}{2\upsilon NT_{s}P_{s}\! +\! \left(NT_{s}\!\!+\!\!\tau\right)N_{0}\!+\!\upsilon
E_{r}T_{s}}\cdot\underbrace{\frac{P_{s}T_{s}}{N_{0}}}_{\varrho}.
\end{align}
Note that,
\begin{align}
\left(\!\sqrt{2E_{a}} \!\!-\!\! \sqrt{E_{b}}\right)^{2}\! \!=\! \!\frac{N_{0}}{\left(\sqrt{2E_{a}}\! \!+\!\! \sqrt{E_{b}}\right)^{2}} \!=\! \frac{1}{\left(\sqrt{\upsilon \varrho\! +\! 1} \!-\! \sqrt{2\upsilon \varrho\!+\! 1}\right)^{2}}\!\sim\! \mathcal{O}\!\left(\frac{1}{\varrho}\right)
\end{align}
thus it can be ignored compared the first term at high SNR region, and $B_{1}$ can be further approximated as
\begin{align}\label{B1}
B_{1}\! =\! \frac{NT_{s}E_{r}}{2\upsilon NT_{s}P_{s} \!+\! \left(NT_{s}\!\!+\!\!\tau\right)N_{0}\!+\!\upsilon T_{s}E_{r}}\cdot\varrho,
\end{align}
where $\varrho$ is defined in \eqref{newEQ1}. Interestingly, the equal power scaling scheme $\gamma_{S}\!=\!\gamma_{I}$ obtains the same $B_{1}$ as \eqref{B1}. Hence, we have the following proposition.
\begin{proposition}
The equal power scaling factors
\begin{align}
\gamma_{S}^{2}\! =\!
\gamma_{I}^{2}\! =\! \frac{P_{r}}{2\tau P_{a} \!+\! \left(NT_{s}\!\!-\!\!\tau\right)P_{b} \!+\! N_{0}},
\end{align}
achieve equal MSE to that of the sub-optimal power allocation in the high SNR region.
\end{proposition}

\subsection{Asymptotic Behavior}
In this subsection, let us discuss the asymptotic behavior with a sufficiently large $N$. With $\gamma_{S} \!=\! \gamma_{I}$ and $\tau\! \ll\! N$, $\delta_{\Sigma}$ can be written as
\begin{align}
\delta_{\Sigma} \!=\!
\frac{\upsilon_{s}\!+\!2\upsilon_{p}C_{1}N}{1\!+\!\upsilon_{s}C_{1}N \!+\!
\upsilon_{p}C_{1}^{2}N^{2}\left[1 \!-\!
|\rho\left(\tau\right)|^{2}\right]}
\end{align}
where $C_{1} \!=\! \frac{P_{s}T_{s}\gamma_{S}^{2}}{N_{0}\left(1 \!+\! \upsilon \gamma_{S}^{2}\right)}$. Note that, when $\left[1 \!-\! |\rho\left(\tau\right)|^{2}\right]$ degrades at a speed lower than $N$, i.e., $\left[1\! -\! |\rho\left(\tau\right)|^{2}\right]\sim\mathcal{O}\left(\frac{1}{N^{i}}\right)$ with $i\! <\! 1$, we have
\begin{align}
\lim_{N\rightarrow\infty}{\delta_{\Sigma}}\rightarrow 0.
\end{align}
Otherwise, there will be an irreducible error. In particular, for two arbitrarily selected sequences, as $N$ goes to infinity, we have
\begin{align}
|\rho\left(\tau\right)|\!\leq\! \left(\frac{NT_{s}\! -\!
\tau}{NT_{s}}\right)^{2}.
\end{align}
When the equality holds, we can obtain
\begin{align}\label{newEQ3}
\lim_{N\rightarrow \infty}{\delta}_{\Sigma}\! =\!
\frac{2\upsilon_{p}}{\upsilon_{s}\! +\! \upsilon_{p}\frac{2\tau
P_{s}T_{s}\gamma_{S}^{2}}{N_{0}\left(1 \!+\!
\upsilon\gamma_{S}^{2}\right)}},
\end{align}
which is an upper bound of the MSE error floor.

Next, let us focus on the high SNR region for a given $\tau$. For the worst case in which $|\rho\left(\tau\right)|\!=\!\left(\frac{NT_{s}\! -\! \tau}{NT_{s}}\right)^{2}$ holds, we have
\begin{align}
\lim_{\varrho\rightarrow \infty}{\delta}_{\Sigma} \!\rightarrow\!
\lim_{\rho \rightarrow
\infty}\mathcal{O}\!\left(\frac{1}{C_{1}}\right) \!=\! 0
\end{align}
for arbitrary non-zero $\tau$. An error floor will occur with $\tau = 0$ as
\begin{align}
\lim_{\varrho\rightarrow \infty}{\delta}_{\Sigma} \rightarrow
\frac{2\upsilon_{p}}{\upsilon_{s}}.
\end{align}
While for the best case in which $|\rho\left(\tau\right)|\!=\!0$, we have $\displaystyle\lim_{\varrho\rightarrow \infty}{\delta}_{\Sigma} \!\rightarrow\! 0$ for arbitrary $\tau$. Hence, we can conclude that for the case in which $\tau \!=\! 0$, the error floor occurs only when the two sequences are fully correlated with each other.

\subsection{Computational Complexity}
To compare the computational complexity of the LMMSE, linear maximum SNR (LMSNR) proposed in \cite{08}, and LMEP estimation algorithms, we calculate the number of complex multiplication operations of these algorithms, which are listed in Table III. It can be seen that the computational complexity of the LMEP algorithm is higher than that of LMMSE and LMSNR. This is because the LMEP requires an initial estimation which consumes a considerable amount of extra computations. However, the complexity of the three algorithms is of same order of $\mathcal{O}\!\left(N^{3}\right)$. This result is not unexpected since all the algorithms are linear solutions, and the dominant computation is in calculating the inversion of a $\left(2N\!+\!1\right)\!\times\!\left(2N\!+\!1\right)$ dimensional matrix.

\begin{table}
\caption{Computational Complexity} \center
\begin{tabular}{c c c c}
Algorithm & Complex Multiplication Times\\
\hline
LMMSE & $\mathcal{O}\!\left[\left(2N\!+\!1\right)^{3}\!+\!4\left(2N\!+\!1\right)^{2}\!+\!2\left(2N\!+\!1\right)\right]$\\
LMSNR & $\mathcal{O}\!\left[\left(2N\!+\!1\right)^{3}\!+\!4\left(2N\!+\!1\right)^{2}\!+\!6\left(2N\!+\!1\right)\right]$\\
LMEP & $\mathcal{O}\!\left[\left(2N\!+\!1\right)^{3}\!+\!8\left(2N\!+\!1\right)^{2}\!+\!8\left(2N\!+\!1\right)\right]$\\\hline
\end{tabular}
\end{table}

\section{Sequence Arriving Order Detection and Scaled LMEP Estimation}
The LMEP estimation algorithm described in the previous section assumes that the SAO is known at all three nodes. In this section, we investigate how to detect the SAO and propose a scaled LMEP estimation algorithm by taking into account the SAO detection error.

Let $\vartheta$ denote the SAO behavior as follows:
\begin{subnumcases}
{\vartheta\!=\!}
0, & signal from $\mathbb{S}_{1}$ arrives prior to that from $\mathbb{S}_{2}$,\nonumber\\
1, & signal from $\mathbb{S}_{2}$ arrives prior to that from
$\mathbb{S}_{1}$.\nonumber
\end{subnumcases}

\subsection{Generalized Likelihood Ratio Testing Approach for SAO Detection}
We consider the SAO detection at $\mathbb{R}$, which can be formulated as a binary hypothesis-testing problem as follows:
\begin{align}
\mathcal{H}_{0}&: \vartheta = 0,\nonumber\\
\mathcal{H}_{1}&: \vartheta = 1. \nonumber
\end{align}
By adopting a similar procedure to that used in obtaining $\mathbf{x}_{S1}$, we can obtain the $\left(2N\!+\!1\right)$ length vector $\mathbf{x}_{R}$ as
\begin{align}
\mathbf{x}_{R}\!=\!\mathbf{T}_{\!\mathcal{H}_{\vartheta}}\!\underbrace{\left[h_{1},
h_{2}\right]^{T}}_{\mathbf{h}} \!+\! \mathbf{w}_{R},
\end{align}
where $\mathbf{T}_{\!\mathcal{H}_{0}} \!=\!\mathbf{\Lambda}\!\left[\mathbf{r}_{1}, \mathbf{r}_{2}\right]$ and $\mathbf{T}_{\!\mathcal{H}_{1}}\! =\! \mathbf{\Lambda}\!\left[\mathbf{r}^{'}_{1}, \mathbf{r}^{'}_{2}\right]$
with
\begin{align}
\mathbf{r}_{1}^{'}
&\!=\!\left[\mathbf{0}_{n_{\tau}},\mathbf{t}_{1}\!\left(1\!:\!N
\!\!-\!\!n_{\tau}\right)^{T}\!\!\otimes\!\mathbf{J},\mathbf{t}_{1}\!\left(N\!
\!-\!\!n_{\tau}\!\!+\!\!1\!:\!N\right)^{T}\right]^{T},\\
\mathbf{r}_{2}^{'} &\!=\!
\left[\mathbf{t}_{2}\!\left(1\!:\!n_{\tau}\right)^{T},\mathbf{t}_{2}\left(n_{\tau}\!\!+\!\!1\!:\!N\right)^{T}\!\!\otimes\!\mathbf{J},\mathbf{0}_{n_{\tau}}\right]^{T}.
\end{align}
Under hypothesis $\mathcal{H}_{\vartheta}$, the likelihood function of $\mathbf{x}_{R}$ is given by
\begin{align}\label{pdf1}
\textsf{p}_{\vartheta}\!\left(\mathbf{x}_{R}|\mathbf{h}\right)\! =\!\left(
\frac{P_{s} T_{s}}{\pi
N_{\!R0}}\right)^{\!2N\!+\!1}\!\!\!\cdot\!\!\!\!\!\!\exp\!\Bigg\{\!\!-\!\frac{P_{s}T_{s}\|\mathbf{x}_{R}\! -\!
\mathbf{T}_{\mathcal{H}_{\vartheta}}\mathbf{h}\|^{2}}{N_{R0}}\!\Bigg\}.
\end{align}
Obviously, $\textsf{p}_{\vartheta}\!\left(\mathbf{x}_{R}|\mathbf{h}\right)$ depends on the instantaneous channel vector $\mathbf{h}$ which is unknown at $\mathbb{R}$. Following the GLRT method, the receiver will decide in favor of hypothesis $\mathcal{H}_{0}$ if
\begin{align}\label{GLRT1}
L_{G}\left(\mathbf{x}_{R}\right) \!=\!
\frac{\textsf{p}_{0}\!\left(\mathbf{x}_{R}|\mathbf{\hat{h}}_{\mathcal{H}_{0}}\right)}{\textsf{p}_{1}\!\left(\mathbf{x}_{R}|\mathbf{\hat{h}}_{\mathcal{H}_{1}}\right)}
\!>\!\lambda_{G} \!=\!
\frac{\textsf{p}\left(\mathcal{H}_{0}\right)}{\textsf{p}\left(\mathcal{H}_{1}\right)},
\end{align}
where $\mathbf{\hat{h}}_{\mathcal{H}_{\vartheta}}$ is the maximum likelihood estimate (MLE) of $\mathbf{h}$ under hypothesis $\mathcal{H}_{\vartheta}$. $\textsf{p}\!\left(\mathcal{H}_{\vartheta}\right)$ is the priori probability of $\mathcal{H}_{\vartheta}$. Here, we set $\textsf{p}\!\left(\mathcal{H}_{0}\right)\! =\!
\textsf{p}\!\left(\mathcal{H}_{1}\right)\! =\! \frac{1}{2}$ leading to a threshold $\lambda_{G}\! =\! 1$. After logarithmic manipulation, \eqref{GLRT1} simplifies to
\begin{align}\label{eq1}
\underbrace{\|\mathbf{x}_{R} \!-\!
\mathbf{T}_{\mathcal{H}_{1}}\mathbf{\hat{h}}_{\mathcal{H}_{1}}\|^{2}
\!-\! \|\mathbf{x}_{R}\! -\!
\mathbf{T}_{\mathcal{H}_{0}}\mathbf{\hat{h}}_{\mathcal{H}_{0}}\|^{2}}_{\Delta_{\mathcal{H}_{1},
\mathcal{H}_{0}}}\!>\! 0,
\end{align}
which indicates that $\mathcal{H}_{0}$ is determined when $\Delta_{\mathcal{H}_{1}, \mathcal{H}_{0}}\! >\! 0$, and $\mathcal{H}_{1}$ is determined with $\Delta_{\mathcal{H}_{1}, \mathcal{H}_{0}}\! <\! 0$. It is widely known that the MLE is equivalent to the least squares estimate (LSE) under Gaussian noise, and thus the MLE of $\mathbf{h}$ under hypothesis $\mathcal{H}_{\vartheta}$ can be written by
\begin{align}
\mathbf{\hat{h}}_{\mathcal{H}_{\vartheta}}\!=\!\left(\mathbf{T}^{H}_{\mathcal{H}_{\vartheta}}
\mathbf{T}_{\mathcal{H}_{\vartheta}}\right)^{\!-1}\mathbf{T}^{H}_{\mathcal{H}_{\vartheta}}\mathbf{x}_{R}.
\end{align}
By substituting $\mathbf{\hat{h}}_{\mathcal{H}_{0}}$ and $\mathbf{\hat{h}}_{\mathcal{H}_{1}}$ in
$\Delta_{\mathcal{H}_{1}, \mathcal{H}_{0}}$ and after some manipulation, we have
\begin{align}
\Delta_{\mathcal{H}_{1}, \mathcal{H}_{0}} \!=\!
\big\|\mathbf{Z}_{\mathcal{H}_{0}}\mathbf{x}_{R}\big\|^{2}
\!-\!\big\|\mathbf{Z}_{\mathcal{H}_{1}}\mathbf{x}_{R}\big\|^{2},
\end{align}
with $\mathbf{Z}_{\mathcal{H}_{\vartheta}}\! =\! \mathbf{T}_{\mathcal{H}_{\vartheta}}\left(\mathbf{T}^{H}_{\mathcal{H}_{\vartheta}}\mathbf{T}_{\mathcal{H}_{\vartheta}}\right)^{-1}\mathbf{T}^{H}_{\mathcal{H}_{\vartheta}}$.
Hence, $\mathcal{H}_{0}$ is determined when $\big\|\mathbf{Z}_{\mathcal{H}_{0}}\mathbf{x}_{R}\big\|^{2}$ is greater than $\big\|\mathbf{Z}_{\mathcal{H}_{1}}\mathbf{x}_{R}\big\|^{2}$, and vice versa. Therefore, the GRLT method in detecting $\vartheta$ simplifies to
\begin{align}
\hat{\vartheta}\! =\! \arg
\max_{\vartheta\!\in\![0,1]}\big\|\mathbf{Z}_{\mathcal{H}_{\vartheta}}\mathbf{x}_{R}\big\|^{2}
\end{align}
We define the equivalent Euclidean distance (EED) between these two hypotheses under hypothesis $\mathcal{H}_{\vartheta}$ by
\begin{align}\label{EQ10}
d_{\mathcal{H}_{\vartheta}}\! =\!
\big\|\mathbf{Z}_{\mathcal{H}_{\vartheta}}\!\left(\mathbf{T}_{\mathcal{H}_{\vartheta}}\mathbf{h}\right)\big\|^{2}
\!-\!\big\|\mathbf{Z}_{\mathcal{H}_{-\vartheta}}\!\left(\mathbf{T}_{\mathcal{H}_{\vartheta}}\mathbf{h}\right)\big\|^{2},
\end{align}
where $\mathcal{H}_{-\vartheta}$ denotes the opposite hypothesis of $\mathcal{H}_{\vartheta}$.
\begin{lemma}
The EED is positive for arbitrary non-zero vector $\mathbf{h}$, and
it can be bounded by
\begin{align}\label{EED1}
d_{\mathcal{H}_{\vartheta}}\!\gtrsim\! \|\mathbf{h}\|^{2}N \left(1 \!-\!
\frac{\left(N \!- \!\frac{\tau}{T_{s}}\right)^{2}}{N^{2} \!-\!
\left(\frac{\lambda}{T_{s}}\right)^{2}}\right).
\end{align}
\end{lemma}
\begin{proof}
See Appendix C.
\end{proof}
This lemma demonstrates that the GLRT solution is effective in detecting $\vartheta$ since the detection is always successful in the absence of noise. The equivalent noise on the EED under hypothesis $\mathcal{H}_{0}$ can be written as
\begin{align}
\tilde{\mathbf{n}}_{\!\mathcal{H}_{0}} \!=&\!
\mathbf{h}^{H}\mathbf{T}^{H}_{\!\mathcal{H}_{0}}\mathbf{n}_{R}\! +\!
\mathbf{n}^{H}_{R}\mathbf{T}_{\!\mathcal{H}_{0}}\mathbf{h}\! +\!
\mathbf{n}^{H}_{R}\mathbf{Z}_{\mathcal{H}_{0}}\mathbf{n}_{R}\nonumber\\&
\!-\!\mathbf{n}^{H}_{R}\mathbf{Z}_{\mathcal{H}_{1}}\mathbf{n}_{R}\!-\!\mathbf{h}^{H}\mathbf{T}^{H}_{\!\mathcal{H}_{0}}\mathbf{Z}_{\mathcal{H}_{1}}\mathbf{n}_{R}
\!-\!\mathbf{n}^{H}_{R}\mathbf{Z}_{\mathcal{H}_{1}}\mathbf{T}_{\!\mathcal{H}_{0}}\mathbf{h}.\nonumber
\end{align}
\begin{lemma}
The equivalent noise on the EED can be approximated as having a normal distribution, and its variance is bounded by
\begin{align}\label{EQ3}
\upsilon_{\mathcal{H}_{\vartheta}}\! \leq\! \underbrace{2\left(2\! -\!
\frac{\tau}{NT_{s}}\right)\frac{\tau}{T_{s}}\|\mathbf{h}\|^{2}\frac{N_{R0}}{P_{s}T_{s}}}_{\upsilon_{E}}.
\end{align}
\end{lemma}
\begin{proof}
See Appendix D.
\end{proof}
By using the two intermediate lemmas, we can calculate the error probability of SAO detection by $P_{\vartheta} \!=\!\mathcal{Q}\left(\!\sqrt{\frac{2d^{2}_{\mathcal{H}_{\vartheta}}}{\upsilon_{\mathcal{H}_{\vartheta}}}}\right)$, and its upper bound is summarized in the following lemma.
\begin{lemma}
Given $\mathbf{h}$ and ${\tau}$, the error probability in detecting $\vartheta$ can be bounded by
\begin{align}\label{EP1}
P_{\vartheta} \!\leq\!
\mathcal{Q}\left(\!\|\mathbf{h}\|\sqrt{\underbrace{\frac{N^{3}\frac{\tau}{T_{s}}
\left(2N \!-\! \frac{\tau}{T_{s}}\right)}{\left[N^{2} \!-\!
\left(\frac{\lambda}{T_{s}}\right)^{2}\right]^{2}}}_{\chi\left(\tau\right)}
\frac{P_{s}T_{s}}{N_{R0}}}\right)
\end{align}
for sufficiently large $N$.
\end{lemma}
\emph{Remarks}:
\begin{itemize}
\item For a fixed $\lambda$ defined in \eqref{newEQ2}, we can specifically write $\chi\!\left(\tau\right)$ as
\begin{align}
\chi\left(\tau\right)\! =\! \frac{N^{3}}{\left(\!N^{2} \!-\!
\left(\frac{\lambda}{T_{s}}\right)^{\!2}\right)^{2}}\left[\!-\!\left(n_{\tau}\!
\!-\! \!N\!\! +\!\! \frac{\lambda}{T_{s}}\right)^{2}\!\! + \!\!\left(2N \!\!-\!\!
\frac{\lambda}{T_{s}}\right)\!\frac{\lambda}{T_{s}} \!+\! \left(N\!\!-\!\!\frac{\lambda}{T_{s}}\right)^{2}\right].
\end{align}
It is clear that $\chi\left(\tau\right)$ increases with respect to $n_{\tau}$. This indicates that $P_{\vartheta}$ declines as $n_{\tau}$ increases. This is because the detection mainly benefits from the un-superimposed part of the pilot signals at the relay and the length of this part increases with increasing ${\tau}$.
\item For a finite $N$, $P_{\vartheta}$ is a decreasing function of $N$. But when $N$ is sufficiently large, we have
\begin{align}
\lim_{N\rightarrow\infty}\chi\left(\tau\right)\rightarrow
\frac{2\tau}{T_{s}}.
\end{align}
This implies that $P_{\vartheta}$ is determined only by $\tau$, and increasing the sequence length of each source's signals will not improve the detection performance.
\item It should be noted that the detection can also be performed at the source node. But it is obvious that the training sequence arriving at the source suffers more severe noise than that arriving at the relay. Thus the error probability of the detection at the source would be much higher than that at relay. Thus, it is recommended to determine $\vartheta$ at the relay first; then the relay can forward the obtained information to both source nodes through a $1$-bit indicator.
\end{itemize}

\subsection{Scaled LMEP Estimation Algorithm}
Since SAO detection is not error free, when $\vartheta$ is erroneous, the MSE of LMMSE initial estimation is no longer equal to $\delta_{\Sigma}$. The corresponding MSE summation becomes
\begin{align}\label{MSE2}
\delta_{\Sigma}^{err}\! =\! \upsilon_{s}\! -\! \frac{Q_{1}}{Q_{0}}\!+\!
\frac{Q_{2}}{Q_{0}^{2}},
\end{align}
where
\begin{align}
Q_{0} \!&=\! 1 \!+\! \upsilon_{s}B_{1}\!+\!\upsilon_{p}B_{1}^{2} \!-\!
\upsilon_{p}|B_{2}|^{2},\nonumber\\
Q_{1} \!&=\! 2\upsilon_{p}B_{1}\big(\mathfrak{R}\{l_{1}\}\upsilon_{a}
\!+\!
\mathfrak{R}\{l_{2}\}\upsilon_{b}\big)\!+\!2\mathfrak{R}\{l_{1}\}\upsilon_{a}^{2}
\!+\! 2\mathfrak{R}\{l_{2}\}\upsilon_{b}^{2},\nonumber\\
 Q_{2} \!&=\!
2\upsilon_{p}^{2}|B_{2}|^{2}\left({|l_{1}|}^{2}\!+\!{|l_{2}|}^{2}\right)
\!+\!\upsilon_{a}\left(\upsilon_{b}B_{1}\!+\!1\right)^{2}|l_{1}|^{2} \!+\!
\upsilon_{b}\left(\upsilon_{a}B_{1}\!+\!1\right)^{2}|l_{2}|^{2}
\nonumber\\&\!+\!B_{1}|B_{2}|^{2}\left[\upsilon_{b}^{2}\left(1\! -\!
\frac{\upsilon_{a}|B_{2}|}{1 \!+\! \upsilon_{a}B_{1}}\right)^{2} \!+\!
\upsilon_{a}^{2}\left(1 \!-\! \frac{\upsilon_{b}|B_{2}|}{1 \!+\!
\upsilon_{b}B_{1}}\right)^{2}\right]\nonumber,
\end{align}
with $l_{i} \!=\! \frac{\gamma_{S}^{2}\mathbf{r}_{i}^{H}\mathbf{\Lambda}^{2}\mathbf{r}_{i}^{'}}{N_{0}\left(1
\!+\! \upsilon\gamma_{S}^{2}\right)}$.

To capture the key parameters affecting the MSE, we focus on the case in which $\rho\left(\tau\right) \!\leq\!\mathcal{O}\!\left(\frac{1}{N}\right)$. Then $\delta_{\Sigma}^{err}$ in this case can be simplified to
\begin{align}
\delta_{\Sigma}^{err}\! = \!\upsilon_{s} \!-\! \frac{Q_{1}}{1 \!+\!
\upsilon_{s}B_{1}\!+\!\upsilon_{p}B_{1}^{2}}\!+\!
\frac{\upsilon_{a}\left(\upsilon_{b}B_{1}\!+\!1\right)^{2}|l_{1}|^{2}\! +\!
\upsilon_{b}\left(\upsilon_{a}B_{1}\!+\!1\right)^{2}|l_{2}|^{2}}{\left(1\!
+\! \upsilon_{s}B_{1}\!+\!\upsilon_{p}B_{1}^{2}\right)^{2}}.
\end{align}
In the high SNR region, we have
\begin{align}
\lim_{\varrho\!\rightarrow \!\infty}{\delta_{\Sigma}^{err}}\! =\! \left\|1 \!-\!
{\frac{l_{1}}{B_{1}}}\right\|^{2}\!\upsilon_{a} \!+\! \left\|1\! -\!
{\frac{l_{2}}{B_{1}}}\right\|^{2}\!\upsilon_{b}.
\end{align}
Since $\|\mathbf{r}_{i}^{H}\mathbf{\Lambda}^{2}\mathbf{r}_{i}^{'}\| \!\leq\! \left(NT_{s} \!-\! \tau\right)P_{S}$ when $\gamma_{S}\! =\!\gamma_{I}$, we further have
\begin{align}\label{EER1}
\lim_{\varrho\!\rightarrow\! \infty}{\delta_{\Sigma}^{err}}\!\geq\!
\left(\frac{\tau}{NT_{s}}\right)^{2}\upsilon_{s}.
\end{align}
Meanwhile, for sufficiently large $N$, $\delta_{\Sigma}^{err}$ gives the same form as \eqref{EER1}. These results imply that the initial estimation with an erroneous $\vartheta$ has an irreducible error floor. Besides, the erroneous $\vartheta$ also affects the effective SNR, which is given by
\begin{align}
\Upsilon^{err} \!=\!
\frac{\mathcal{E}\Big\{|\hat{h}_{b}|^{2}\Big\}}{\mathcal{E}\Big\{|\hat{h}_{a}|^{2} \!+\!
|\hat{h}_{b} \!-\! h_{b}|^{2}\Big\} \!+\! |h_{a}|^{2}\!+\!
\frac{N_{0}}{P_{s}T_{s}}\!\left(|h_{a}|\! +\!
\frac{1}{\alpha^{2}}\right)}.
\end{align}
Apparently, $\Upsilon^{err} < \Upsilon$ always holds, leading to an increased error probability. All these results indicate that a severe SAO detection error would substantially degrade the estimation performance.

To compensate for this performance degradation, we propose a scaled LMEP estimator to minimize the average BEP written by
\begin{align}\label{EQ9}
\bar{P}_{E} \!= \left(1 \!-\! P_{\vartheta}\right)\!\mathcal{Q}\left(\!\sqrt{\Upsilon}\right)\!+\!
P_{\vartheta}\!\cdot\!\mathcal{Q}\left(\!\sqrt{\Upsilon_{err}}\right).
\end{align}
\begin{lemma}
When $0\!<\!P_{\vartheta}\!\leq\!\frac{1}{2}$, the optimal estimator $\mathfrak{u}^{*}\!=\!\big\{\mathbf{u}_{a}^{*}, \mathbf{u}_{b}^{*}\big\}$ can be obtained by solving the following optimization:
\begin{align}\label{EQ7}
\max_{\mathfrak{u}} & \quad\Upsilon\left(\mathfrak{u}\right) \!+ \! \Upsilon_{err}\left(\mathfrak{u}\right)\\
s.t. & \quad\Upsilon\left(\mathfrak{u}\right)\! -\!
\Upsilon_{err}\left(\mathfrak{u}\right)\! =\! 2\log{\frac{1 \!-\!
P_{\vartheta}}{P_{\vartheta}}}.
\end{align}
\end{lemma}
\begin{proof}
Let $\mathfrak{u}\!=\!\{\mathbf{u}_{a}, \mathbf{u}_{b}\}$ denote a feasible linear estimator. Since $\Upsilon\!\left(\mathfrak{u}\right)$ and $\Upsilon_{err}\!\left(\mathfrak{u}\right)$ are both concave functions of $\mathfrak{u}$, $\bar{P}_{E}$ is convex with respect to $\mathfrak{u}$. Thus the optimal linear estimators can be obtained by convex optimization methods. By using the Chernoff bound \cite{20}, the computed BEP can be approximated as
\begin{align}
\bar{P}_{E} \!\approx\! \frac{\left(1\! -\!P_{\vartheta}\right)}{2}\exp\!\Big\{\!\!-\!\!\frac{1}{2}\Upsilon\!\left(\mathfrak{u}\right)\!\Big\}\!+\!
\frac{P_{\vartheta}}{2}\!\exp\!\Big\{\!\!-\!\!\frac{1}{2}\Upsilon_{err}\!\left(\mathfrak{u}\right)\!\Big\}.
\end{align}
Let $\mathfrak{u}^{+}$ and $\mathfrak{u}^{-}$ denote the values of $\mathfrak{u}$ that maximize $\Upsilon$ and $\Upsilon^{err}$, respectively. The $\mathfrak{u}$ minimizing $\bar{P}_{E}$, denoted by $\mathfrak{u}^{*}$, should satisfy the following conditions:
\begin{align}
&\Upsilon\left(\mathfrak{u}^{-}\right)\!\leq\!\Upsilon\!\left(\mathfrak{u}^{*}\right)\!\leq\!
\Upsilon\!\left(\mathfrak{u}^{+}\right),\\
&\Upsilon_{err}\!\left(\mathfrak{u}^{+}\right)\!\leq\!\Upsilon_{err}\!\left(\mathfrak{u}^{*}\right)\!
\leq\!\Upsilon_{err}\!\left(\mathfrak{u}^{-}\right).
\end{align}
By defining $\Upsilon_{m}\!\left(\cdot\right)\! =\!\Upsilon\!\left(\cdot\right)\!-\!\Upsilon_{err}\!\left(\cdot\right)$ and $\Upsilon_{n}\!\left(\cdot\right)
\!=\!\Upsilon\!\left(\cdot\right)\!+\!\Upsilon_{err}\!\left(\cdot\right)$, $\bar{P}_{E}$ can be written as
\begin{align}
\bar{P}_{E}\! =\! \frac{\left(1 \!-\!P_{\vartheta}\right)}{2}\!\exp\!\Big\{\!\!-\!\!\frac{1}{4}\left(\Upsilon_{n}\! +\!
\Upsilon_{m}\right)\Big\}\!+\!
\frac{P_{\vartheta}}{2}\!\exp\!\Big\{\!\!-\!\!\frac{1}{4}\left(\Upsilon_{n}\! -\!
\Upsilon_{m}\right)\Big\}.
\end{align}
By treating $\bar{P}_{E}$ as a function of $\Upsilon_{m}$, we have
\begin{align}
\frac{\partial \bar{P}_{E}}{\partial \Upsilon_{m}} \!=\! \frac{1}{8}\exp\!\Big\{\!\!-\!\!\frac{1}{4}\Upsilon_{n}\!\Big\}\!\left[P_{\vartheta} \exp\!\Big\{\frac{1}{4}\Upsilon_{m}\Big\} \!-\!\left(1 \!\!-\!\! P_{\vartheta}\right)\exp\Big\{\!\!-\!\!\frac{1}{4}\Upsilon_{m}\Big\}\right].
\end{align}
By solving $\frac{\partial \bar{P}_{E}}{\partial \Upsilon_{m}}\! = \!0$, we have
\begin{align}\label{Equality}
\Upsilon_{m}\!=\!2\log\!\frac{1\!-\!P_{\vartheta}}{P_{\vartheta}}.
\end{align}
Since $\frac{\partial^{2} \bar{P}_{E}}{\partial \Upsilon_{m}^{2}} \!>\! 0$, $\bar{P}_{E}$ achieves its global minimum when \eqref{Equality} holds. Thus $\mathfrak{u}^{*}$ should satisfy $\Upsilon_{m}\left(\mathfrak{u}^{*}\right)\!=\!2\log\!\frac{1\!-\!P_{\vartheta}}{P_{\vartheta}}$. Meanwhile, it can be easily proved that $\bar{P}_{E}$ is decreasing with respect to $\Upsilon_{n}$, and thus $\Upsilon_{n}\!\left(\mathfrak{u}^{*}\right)\! \geq\! \Upsilon_{n}\!\left(\mathfrak{u}\right)$ always holds for an arbitrary $\mathfrak{u}$ that satisfies $\Upsilon_{m}\!\left(\mathfrak{u}^{*}\right)\!=\!2\log\!\frac{1\!-\!P_{\vartheta}}{P_{\vartheta}}$, i.e., $\Upsilon\!\left(\mathfrak{u}^{*}\right) \!-\!
\Upsilon_{err}\!\left(\mathfrak{u}^{*}\right) \!=\! 2\log{\frac{1 \!-\!
P_{\vartheta}}{P_{\vartheta}}}$. The proof is complete.
\end{proof}
Note that, it is hard to obtain a closed-form expression of $\mathfrak{u}^{*}$ due to the complex expressions for $\Upsilon\!\left(\mathfrak{u}\right)$ and $\Upsilon_{err}\!\left(\mathfrak{u}\right)$. However, it is clear that \eqref{EQ7} is a concave optimization problem and the constraint is an equality. Thus we can use numerical methods to effectively obtain the optimal solution with very low computational complexity.

The SLMEP algorithm is summarized in Table IV. In the first step, we need to obtain the initial estimates of $h_{a}$ and $h_{b}$ using LMMSE estimators for both $\vartheta \!=\! 0$ and $\vartheta \!=\! 1$. Then in the second step, the estimates are then updated with the SLMEP estimator.

\begin{table}
\caption{SLMEP Estimation Algorithm} \center
\begin{tabular}{p{5in}@{}}
\hline
\begin{itemize}
\item \textbf{Step 1}
\begin{itemize}
\item[-] Construct the LMMSE estimator $\mathbf{D}$ with $[\mathbf{r}_{1},\mathbf{r}_{2}]$.
\item[-] Obtain the initial estimates of $h_{a}$ and $h_{b}$ for $\vartheta = 0$ as $[\hat{h}_{a},\hat{h}_{b}]^{T} = \mathbf{D}\mathbf{x}_{S1}$.
\item[-] Construct the LMMSE estimator $\mathbf{D}^{'}$ with $[\mathbf{r}_{1}^{'},\mathbf{r}_{2}^{'}]$.
\item[-] Obtain the initial estimates of $h_{a}$ and $h_{b}$ for $\vartheta = 1$ as $[\hat{h}_{a}^{'},\hat{h}_{b}^{'}]^{T} = \mathbf{D}^{'}\mathbf{x}_{S1}$.
\end{itemize}
\item \textbf{Step 2}
\begin{itemize}
\item[-] Construct $\Upsilon(\mathfrak{u})$ and $\Upsilon_{err}(\mathfrak{u})$ with $[\hat{h}_{a},\hat{h}_{b}]$ and $[\hat{h}_{a}^{'},\hat{h}_{b}^{'}]$, respectively.
\item[-] Obtain the suboptimal estimator $\mathfrak{u}^{*}=[\mathbf{u}_{a}^{*}, \mathbf{u}_{b}^{*}]$ by \eqref{EQ7}
\item[-] Update the estimate of $h_{a}$ and $h_{b}$ as $\hat{h}_{a} =
(\mathbf{u}_{a}^{*})^{H}\mathbf{x}_{S1}$ and $\hat{h}_{b} =
(\mathbf{u}_{b}^{*})^{H}\mathbf{x}_{S1}$.
\end{itemize}
\item \textbf{Return} $\mathbf{\hat{\Theta}} = [\hat{h}_{a},\hat{h}_{b}]^{T}$.
\end{itemize}
\\
\hline
\end{tabular}
\end{table}

\section{Performance Evaluation}
In this section, we investigate the performance of the proposed estimation algorithm in the presence of synchronization error in the TWRN. The channel coefficients $h_{1}$ and $h_{2}$ are generated by independent circularly symmetric complex Gaussian random variables with zero mean and unit variances. The noise at either relay or source is assumed to be AWGN with unit variance. The symbol period $T_{s}$ is set to be of unit value, and thus the transmit power of each source node is $P_{s}$, and the total power consumption of the relay is $E_{r} \!=\! NP_{s}$. The average SNR $\varrho$ is equal to $P_{s}$ with $N_{0}\!=\!N_{R0}\!=\!N_{S0}\!=\!1$.

In Fig. \ref{fig_1}, we compare the error performance of the proposed LMEP algorithm with LMMSE estimation and the LMSNR estimation technique presented in \cite{08}. We consider a binary phase shift keying (BPSK) modulation and assume that the synchronization error $\tau$ satisfies a uniform distribution within $\left[0,NT_{s}\right]$. It can be seen that the proposed LMEP algorithm with LMMSE initialization achieves lower BER than that of both LMMSE and LMSNR estimation algorithms for a given $N$, which demonstrates the advantage of the proposed algorithm. Further, we observe that the LMEP algorithm with random initialization endures very high BER and performs much worse compared with other methods. This implies that the error performance of the LMEP algorithm differs significantly with different initializations and an inappropriately selected stating point could negatively affect the error performance.

\begin{figure}[!h]
\center
  \includegraphics[width=0.5\textwidth]{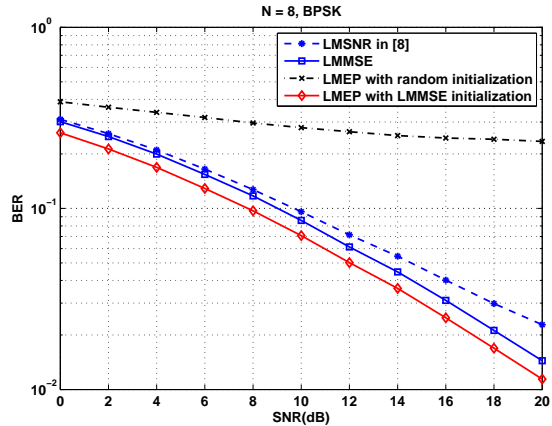}
 \caption{Simulated bit error rate versus SNR for different estimation methods.} \label{fig_1}
\end{figure}

Next let us investigate the MSE estimation of $h_{a}$ and $h_{b}$ with LMMSE initial estimation by using the following different training sequences:
\begin{itemize}
\item Type 1: $\mathbf{t}_{1}\! =\! [\underbrace{1,\ldots,1}_{N}]^{T}$ and $\mathbf{t}_{2}\! =\! [\underbrace{1,\ldots,1}_{N/2},\underbrace{-1,\ldots,-1}_{N/2}]^{T}$.
\item Type 2: The 3rd and 4th column of the $N\times N$ DFT matrix are selected as $\mathbf{t}_{1}$ and $\mathbf{t}_{2}$, respectively.
\item Optimal: Following Proposition 2, the 1st and $\left(\frac{N}{2}\!+\!1\right)$th column of the $N\times N$ DFT matrix are selected as $\mathbf{t}_{1}$ and $\mathbf{t}_{2}$, respectively.
\end{itemize}
The average MSE versus $\tau$ with different types of training sequences for $N\!=\!8,16$ is shown in Fig. \ref{fig_2}. $\varrho$ is set to be $10$dB. It can be observed that the optimal training achieves the lowest MSE among all the three types of training sequences for $\tau\!\in\![0,NT_{s}]$, and the resulting MSE with optimal training almost remains the same for different $\tau$. This is in contrast to Type 1 and 2, for which the MSE varies significantly with $\tau$. This is because the resulting MSE is mainly determined by $|\rho\!\left(\tau\right)|$, and the optimal training aims at minimizing the maximum $|\rho\left(\tau\right)|$ for all $\tau$ values.

\begin{figure}[!h]
\center
  \includegraphics[width=0.5\textwidth]{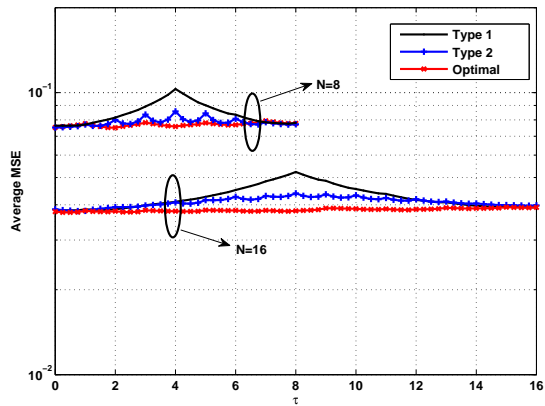}
 \caption{MSE versus $\tau$ for different training sequence selections.} \label{fig_2}
\end{figure}

Next let us compare the MSE performance of channel estimation with optimal training sequences to that with training sequence consisting of $N$ random quadrature-phase-shift-keying (QPSK) symbols. In accordance with the curves in Fig. \ref{fig_2_2}, we can see that the MSE with optimal training sequences is considerably lower than that constructed of random QPSK symbols. This is because the latter cannot ensure the low correlation between the two training sequences, which would negatively affect the estimation performance, leading to an increased MSE.

\begin{figure}[!h]
\center
  \includegraphics[width=0.5\textwidth]{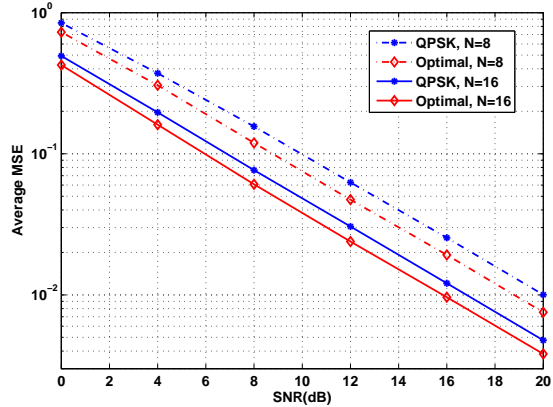}
 \caption{MSE versus SNR for different training sequence selections.} \label{fig_2_2}
\end{figure}

Fig.\ref{fig_3} illustrates the MSEs versus SNR for different power allocation schemes. Here, we assume that $\tau$ follows a uniform distribution within $\left[0,NT_{s}\right]$. We compare the proposed sub-optimal pilot power allocation (SOA) with the random power allocation (RA), where $\gamma_{I}^{2}$ is uniformly selected within $\left[0, \frac{E_{r}}{\left(NT_{s}\! -\! {\tau}\right)P_{b}}\right]$, and the equal power allocation (EA) where $\gamma_{S} \!=\! \gamma_{I}$. We can observe that the SOA and EA achieve lower MSE than that of RA. Fig. \ref{fig_3_1} compares the BER performance of the three power allocation schemes. It can be seen that the SOA scheme performs the best among all the three power allocation schemes. We can also see that SOA and EA achieve almost the same performance over a wide range of SNR. This result validates the statement presented in Proposition 3. In Fig.\ref{newFig}, we investigate the MSE performance versus the training sequence length $N$, and the parameters are set as SNR$=0$dB, $\tau\!=\!0.5T_{s}$. The optimal training refers to the sequence derived in proposition 3, while the worst training refers to the simulation in which $\mathbf{t}_{1}$ and $\mathbf{t}_{2}$ are fully correlated. It can be noted from the figure that the average MSE performance with the optimal training decreases quickly as $N$ increases. With the worst training sequence, an MSE error floor will occur and the corresponding MSE converges to the asymptotic result in \eqref{newEQ3} for a sufficiently large $N$($>40$). This validates the consistency between analysis and numerical results.

\begin{figure}[!h]
\center
  \includegraphics[width=0.5\textwidth]{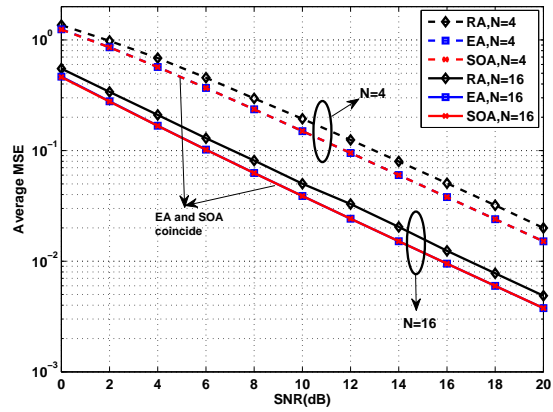}
 \caption{MSE versus SNR for different power allocation schemes.} \label{fig_3}
\end{figure}

\begin{figure}[!h]
\center
  \includegraphics[width=0.5\textwidth]{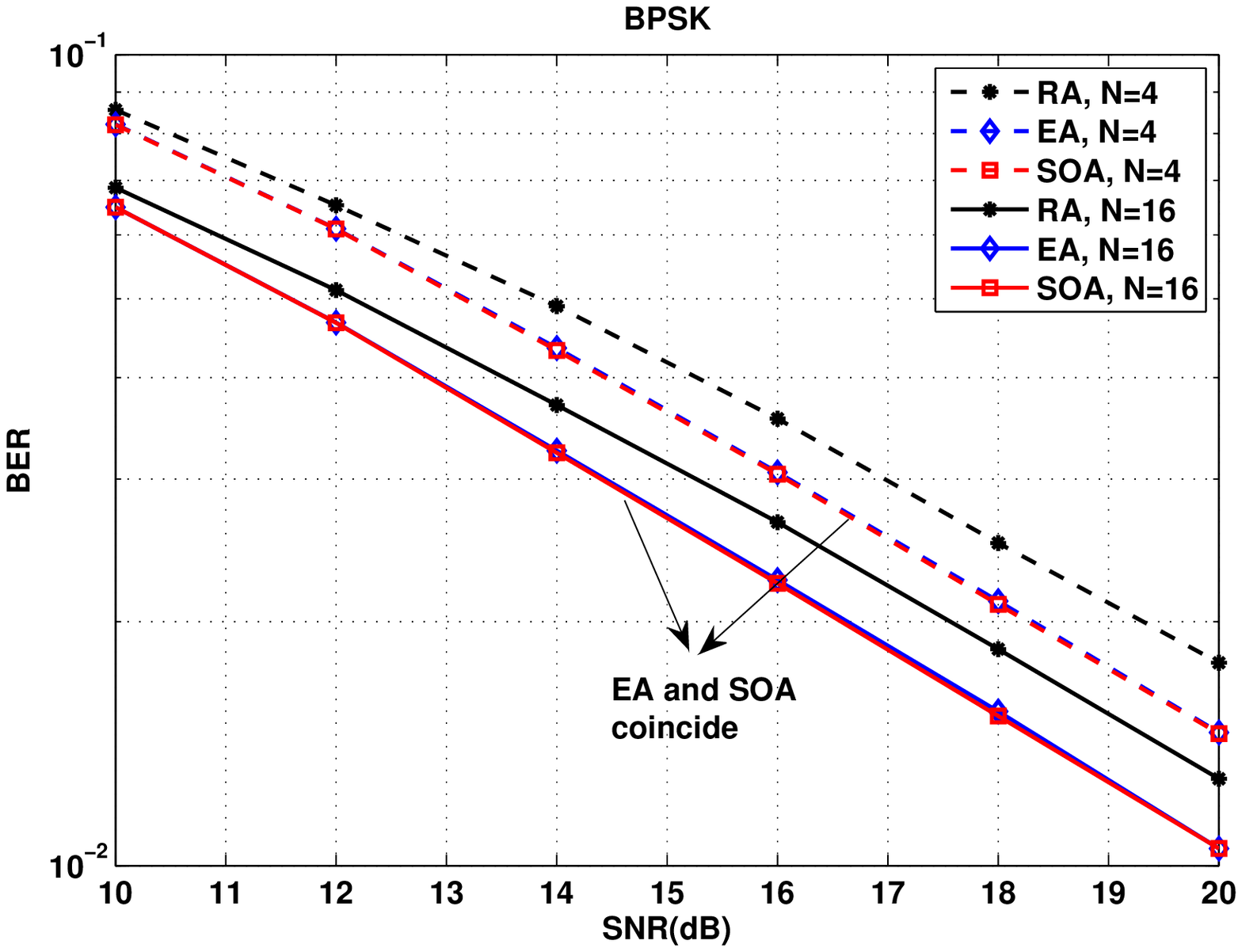}
 \caption{BER versus SNR for different power allocation schemes.} \label{fig_3_1}
\end{figure}

\begin{figure}[!h]
\center
  \includegraphics[width=0.5\textwidth]{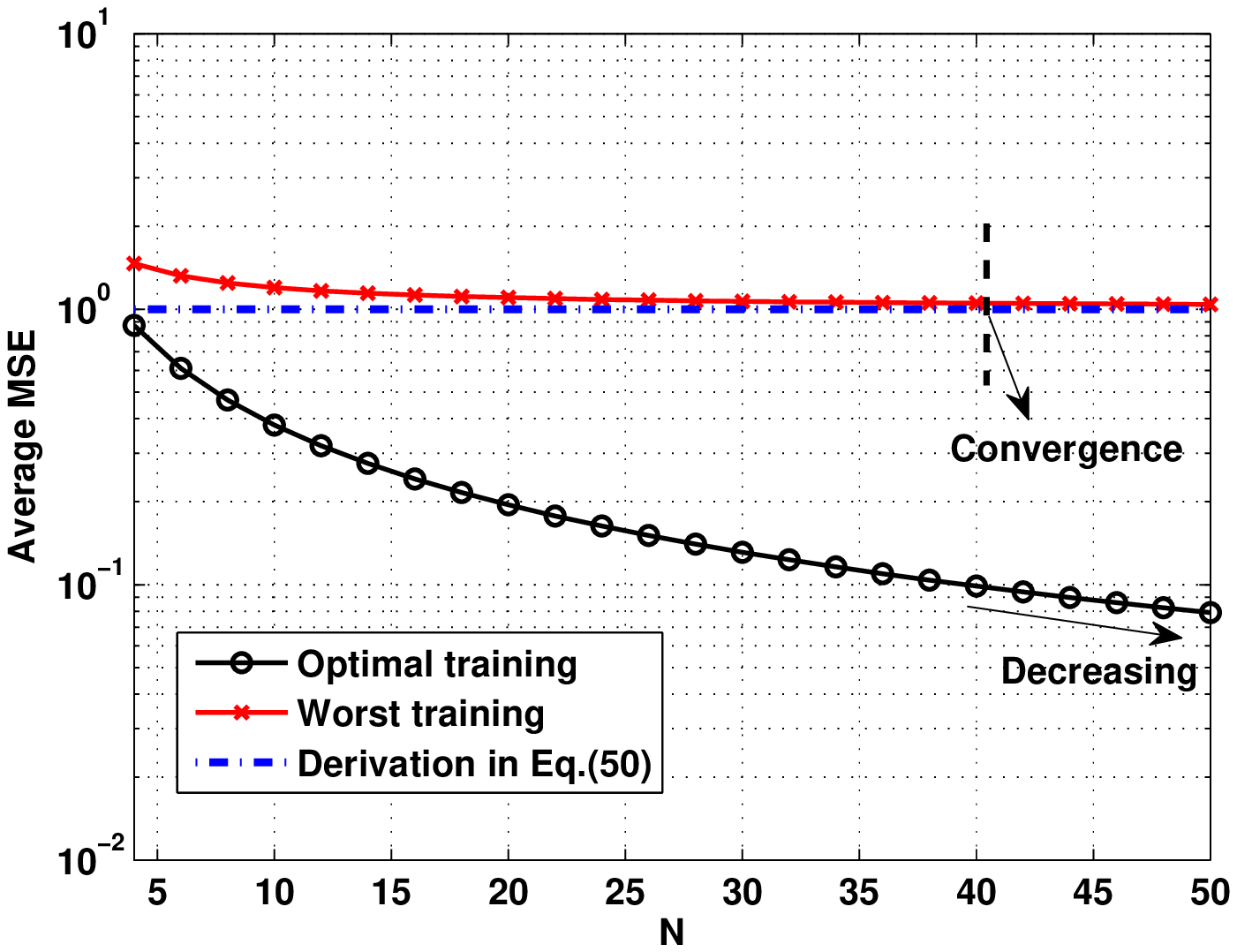}
 \caption{MSE versus $N$ for different training sequences.} \label{newFig}
\end{figure}

Next, let us evaluate the error probability of SAO detection denoted by $P_{\vartheta}$ with joint estimation of $h_{a}$ and $h_{b}$ by using the GLRT method. The error probabilities of SAO detection at the relay and at the source are both displayed versus $\tau$ in Fig.\ref{fig_4}. It can be noted that $P_{\vartheta}$ quickly decreases with SNR, which demonstrates the effectiveness of the GLRT solution. Recalling \eqref{EP1}, it can be easily seen that $P_{\vartheta}$ decreases with respect to $\tau$ when $\lambda = 0$. When focusing on the points $\tau \!=\! n_{\tau}T_{s}$, we can see that $P_{\vartheta}$ monotonously decreases as $n_{\tau}$ increases. This is consistent with the analytical result. We can also see that $P_{\vartheta}$ also decreases with respect to $n_{\tau}$ for a given $\lambda$. This is also consistent with the upper bound of $P_{\vartheta}$ in \eqref{EP1}.

\begin{figure}[!h]
\center
  \includegraphics[width=0.5\textwidth]{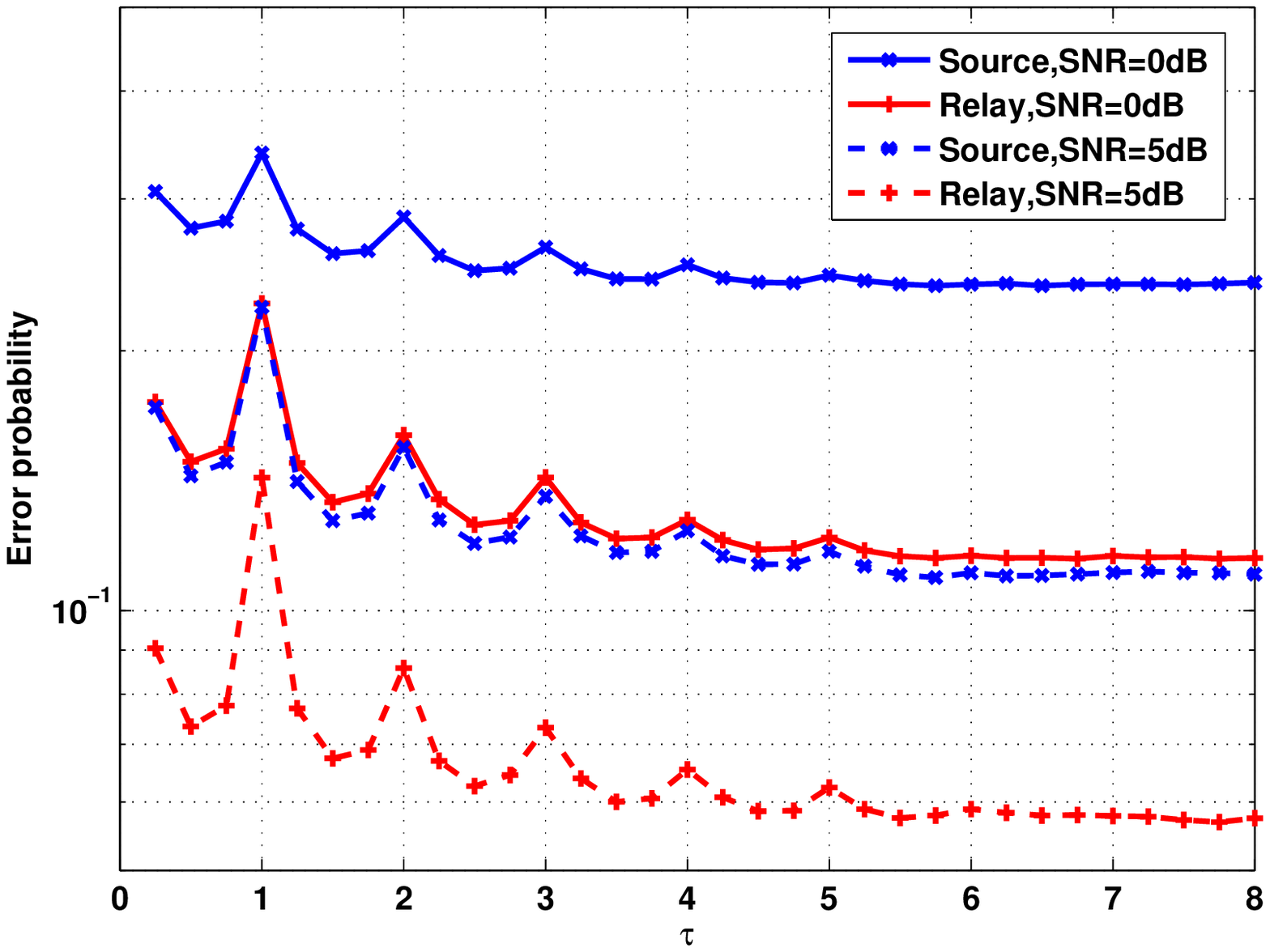}
 \caption{Error probability of SAO detection versus $\tau$.} \label{fig_4}
\end{figure}

Finally, we evaluate the performance of the SLMEP algorithm and compare it with the LMEP solution with the practical SAO detection. In Fig.\ref{fig_5}, we display the BER versus $P_{\vartheta}$ with fixed $\tau\! =\! 1$. It can be noted that the BER of LMEP increases with respect to $P_{\vartheta}$. This indicates that the SAO detection error significantly degrades the system performance. We can also observe that the SLMEP algorithm attains almost the same BER for different $P_{\vartheta}$ and achieves lower BER than that of the LMEP method when $P_{\vartheta}$ is greater than a threshold, but the performance gain degrades as $P_{\vartheta}$ decreases. This corroborates that the SLMEP estimation can compensate for the performance degradation caused by SAO detection error, especially when $P_{\vartheta}$ is high.

\begin{figure}[!h]
\center
  \includegraphics[width=0.5\textwidth]{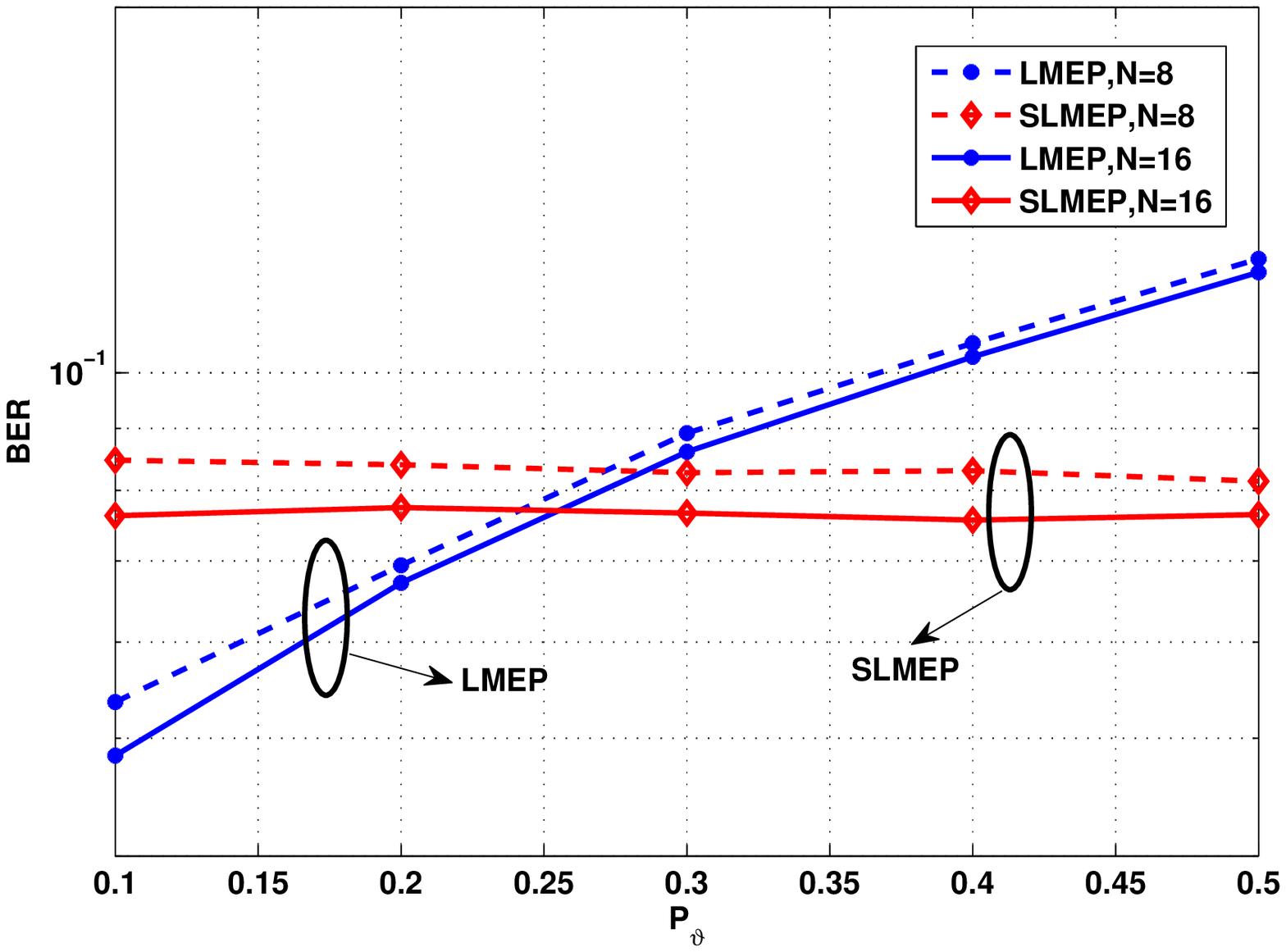}
 \caption{Bit error rate versus $P_{\vartheta}$ for different estimation algorithms.} \label{fig_5}
\end{figure}

\section{Conclusions}
In this paper, we have investigated the channel estimation problem in TWRNs for which the two source nodes are not perfectly synchronized with each other. We have proposed a two-step estimation algorithm with initial LMMSE channel estimation, followed by LMEP estimation to minimize the error probability of coherent reception. We have also designed optimal training sequences and power allocation at the relay to further improve the LMMSE initial estimation. We have further presented the GLRT method to overcome the SAO ambiguity, and an upper bound on its error probability has been derived. A scaled LMEP estimation algorithm has also been proposed to compensate for the degradation of channel estimation performance caused by SAO detection error. Simulation results have demonstrated that the proposed algorithm can effectively overcome the negative effects of synchronization error and SAO ambiguity, and significantly improve the system performance compared to the existing LMMSE and LMSNR estimation algorithms.

\appendices
\section{Proof of Proposition 1}
Let us select the $k_{i}$-th column of the ${N\!\times\! N}$ DFT matrix as $\mathbf{t}_{i}$, i.e., $\mathbf{t}_{i}\!\! =\!\! \left[1,e^{\!-\!j\phi_{i}},\ldots,e^{\!-\!j\left(N \!-\! 1\right)\phi_{i}}\right]^{T}$ with $\phi_{i}\!=\!\frac{2\pi}{N}\left(k_{i}-1\right)$. The magnitude of $\rho$ for a given ${\tau}$ is given by
\begin{align}
|\rho\left({\tau}\right)| \!&=\! \frac{1}{N}\Bigg|\frac{1\! -\!
e^{j\left(\phi_{1} \!-\! \phi_{2}\right)\left(N \!-\! n_{\tau}\right)}}{1 \!-\!
e^{j\left(\phi_{1} \!-\!\phi_{2}\right)}}e^{jn_{\tau}\phi_{1}}\frac{T_{s}\!-\!\lambda}{T_{s}}\!+\!
\frac{1 \!-\! e^{j\left(\phi_{1} \!-\! \phi_{2}\right)\left(N \!-\! n_{\tau} \!-\!1\right)}}{1 \!-\! e^{j\left(\phi_{1} \!-\!
\phi_{2}\right)}}e^{j\left(n_{\tau}\!+\!1\right)\phi_{1}}\frac{\lambda}{T_{s}}
\Bigg|.\nonumber\\
\!&\leq\! \frac{1}{N}\bigg\{\frac{T_{s}\! -\! \lambda}{T_{s}}\Big|\frac{1 \!-\!
e^{j\left(\phi_{1} \!-\! \phi_{2}\right)\left(N \!-\! n_{\tau}\right)}}{1 \!-\!
e^{j\left(\phi_{1}\! -\! \phi_{2}\right)}}\Big| \!+\!
\frac{\lambda}{T_{s}}\Big|\frac{1 \!-\! e^{j\left(\phi_{1} \!-\!
\phi_{2}\right)\left(N \!-\! n_{\tau} \!-\!1 \right)}}{1\! -\!
e^{j\left(\phi_{1} \!-\! \phi_{2}\right)}}\Big|\bigg\}.\nonumber
\end{align}
With a sufficiently large $N$ and a finite $\tau$, we have
\begin{align}\nonumber
\Big|\frac{1\! -\! e^{j\left(\phi_{1} \!-\! \phi_{2}\right)\left(N \!-\! n_{\tau}\right)}}{1 \!-\! e^{j\left(\phi_{1} \!-\! \phi_{2}\right)}}\Big| \!\leq\! 1, \Big|\frac{1\! -\! e^{j\left(\phi_{1}\! -\! \phi_{2}\right)\left(N\! -\! n_{\tau} \!-\! 1\right)}}{1 \!-\! e^{j\left(\phi_{1}\! -\! \phi_{2}\right)}}\Big| \!\leq\! 1
\end{align}
which leads to
\begin{align}\nonumber
|\rho\left(\tau\right)|\!\leq\! \frac{1}{N}.
\end{align}
Furthermore, in the limit we have
\begin{align}\nonumber
\lim_{N\rightarrow\infty}|\rho\left(\tau\right)|\!=\!
\lim_{N\rightarrow\infty}\mathcal{O}\!\left(\frac{1}{N}\right)\! =\!0.
\end{align}
This completes the proof.

\section{Proof of Proposition 2}
Recalling Appendix A, for arbitrary $k_{1}$ and $k_{2}$,
$|\rho\left(\tau\right)|$ satisfies
\begin{equation}
\min \max_{{\tau}\in\left[0, NT_{s}\right]}|\rho\left({\tau}\right)| \!\geq\! |\rho\left[\left(N \!-\!1\right)T_{s}\right]| \!=\! \frac{1}{N}.\nonumber
\end{equation}
Thus we have
\begin{align}
\max_{{\tau}\in\left[0, NT_{s}\right]} |\rho\left({\tau}\right)|
\!\geq\! \frac{1}{N}.\nonumber
\end{align}
Focusing on the case in which $N$ is even, $|\rho\left({\tau}\right)|$
can be simplified by substituting $k_{2} - k_{1} = \frac{N}{2}$ into
it to yield
\begin{align}
|\rho\left({\tau}\right)|\!&=\!\frac{1}{2N}\Bigg|\frac{T_{s} \!-\!\lambda}{T_{s}}{\left[1\! -\! e^{j\left(\phi_{1}\! -\!\phi_{2}\right)\left(N\! -\! n_{\tau}\right)}\right]} \!+\!
\frac{\lambda}{T_{s}}{\left[1\! -\! e^{j\left(\phi_{1}\! -\!\phi_{2}\right)\left(N \!-\! n_{\tau} \!-\!1\right)}\right]}\Bigg|\nonumber\\
\!&\leq\!\frac{1}{2N}\bigg\{\frac{T_{s} \!-\! \lambda}{T_{s}}{\Big|1\! -\!e^{j\left(\phi_{1}\! -\! \phi_{2}\right)\left(N \!-\! n_{\tau}\right)}\Big|}\!+\! \frac{\lambda}{T_{s}}{\Big|1 \!-\!e^{j\left(\phi_{1}\! -\!\phi_{2}\right)\left(N \!-\! n_{\tau}\! -\!1\right)}\Big|}\bigg\}\!\leq\! \frac{1}{N}.\nonumber
\end{align}
In this case, the maximum value of $|\rho\left(\tau\right)|$
satisfies
\begin{align}
\max_{{\tau}\in\left[0, NT_{s}\right]} |\rho\left({\tau}\right)|
\!\leq\! \frac{1}{N},\nonumber
\end{align}
which indicates that $k_{2}\!-\!k_{1} \!=\! \frac{N}{2}$ is the optimal
design when $N$ is even. For the case in which $N$ is odd, substituting
$k_{2}\! - \!k_{1} \!=\! \frac{N\pm1}{2}$ into $|\rho\left(\tau\right)|$ yields
\begin{align}
|\rho\left({\tau}\right)|\!\leq\! \frac{1}{N}\bigg\{\frac{T_{s} \!-\!\lambda}{T_{s}}\Big|\frac{1\! -\! e^{\!-\!j\pi n_{\tau} \pm j \pi
\frac{n_{\tau}}{N}}}{1 \!+\! e^{\pm j\pi\frac{1}{N}}}\Big| \!+\!
\frac{\lambda}{T_{s}}\Big|\frac{1 \!-\! e^{-j\pi n_{\tau} \pm j \pi
\frac{n_{\tau}\!+\!1}{N}}}{1 \!+\! e^{\pm j\pi\frac{1}{N}}}\Big|\bigg\}.
\nonumber
\end{align}
When $n_{\tau}$ is odd, we have
\begin{align}
|\rho\left({\tau}\right)|\!\leq\!\frac{1}{N}\left[\frac{T_{s}\! -\!
\lambda}{T_{s}}\Big|\frac{1 \!+\! e^{\pm j \pi \frac{n_{\tau}}{N}}}{1 \!+\!
e^{\pm j\pi\frac{1}{N}}}\Big| \!+\! \frac{\lambda}{T_{s}}\Big|\frac{1 \!+\!
e^{\pm j \pi \frac{n_{\tau}\!+\!1}{N}}}{1 \!+\! e^{\pm
j\pi\frac{1}{N}}}\Big|\right]\!\leq\! |\rho\left({T_{s}}\right)|\!=\!\frac{1}{N}.\nonumber
\end{align}
When $n_{\tau}$ is even, we have
\begin{align}
|\rho\left({\tau}\right)|\!\leq\!\frac{1}{N}\left[\frac{T_{s}\! -\!\lambda}{T_{s}}\Big|\frac{1 \!-\! e^{\pm j \pi \frac{n_{\tau}}{N}}}{1 \!+\!e^{\pm j\pi\frac{1}{N}}}\Big| \!+\! \frac{\lambda}{T_{s}}\Big|\frac{1 \!-\!e^{\pm j \pi \frac{n_{\tau}\!+\!1}{N}}}{1 \!+\! e^{\pm
j\pi\frac{1}{N}}}\Big|\right]\!\leq\! |\rho\left(\left(N \!-\!1\right)T_{s}\right)|\!=\!\frac{1}{N}.\nonumber
\end{align}
Then, $|\rho\left(\tau\right)|$ satisfies
\begin{align}
\max_{{\tau}\in\left[0, NT_{s}\right]} |\rho\left({\tau}\right)|
\!\leq\! \frac{1}{N}\nonumber
\end{align}
for arbitrary $\tau$. Therefore, we can conclude that $k_{2} \!-\! k_{1}\!
=\! \frac{N\pm1}{2}$ is the optimal design for odd $N$. This completes the proof.

\section{Proof of Lemma 1}
Let us selecting the $k_{i}$-th column of the ${N\!\times\! N}$ DFT matrix as $\mathbf{t}_{i}$, i.e., $\mathbf{t}_{i} \!=\! \left[1,e^{\!-\!j\phi_{i}},\ldots,e^{\!-\!j\left(N \!-\!
1\right)\phi_{i}}\right]^{T}$ with $\phi_{i}= \frac{2\pi}{N}\left(k_{i}-1\right)$. After some mathematical calculations, the term $\mathbf{T}^{H}_{\mathcal{H}_{0}}\mathbf{T}_{\mathcal{H}_{0}}$ with a given ${\tau}$ can be calculated as
\begin{align}\nonumber
\mathbf{T}^{H}_{\mathcal{H}_{0}}\mathbf{T}_{\mathcal{H}_{0}}\! =\!
\left[\begin{array}{cc} N & \Psi_{1}\left(\tau\right) \\
\Psi_{1}^{*}\left(\tau\right) & N\end{array}\right],
\end{align}
where
\begin{align}
\Psi_{1}\!\left(\tau\right) \!=\! \frac{T_{s} \!-\!\lambda}{T_{s}}e^{j{n_{\tau}}\phi_{1}}\sum_{i = 1}^{N \!-\!n_{\tau}}{e^{j\left(\phi_{1} \!-\! \phi_{2} \right)\left(i \!- \!1\right)}}
\!+\! \frac{\lambda}{T_{s}}e^{j\left(n_{\tau} \!+\! 1\right)\phi_{1}}\sum_{i
= 1}^{N \!-\!n_{\tau} \!-\! 1}{e^{j\left(\phi_{1}\! -\! \phi_{2} \right)\left(i\!-\!1\right)}}.\nonumber
\end{align}
And $\mathbf{T}^{H}_{\mathcal{H}_{0}}\mathbf{T}_{\mathcal{H}_{1}}$
can be calculated by
\begin{align}\nonumber
\mathbf{T}^{H}_{\mathcal{H}_{0}}\mathbf{T}_{\mathcal{H}_{1}} \!=\!
\left[\begin{array}{cc} \Psi_{2}\left(\tau,\phi_{1}\right) & 0\\0 &
\Psi_{2}\left(\tau,-\phi_{2}\right)\end{array}\right],
\end{align}
where
\begin{align}\nonumber
\Psi_{2}\!\left(\tau, \phi\right) \!=\! \frac{T_{s} \!-\!
\lambda}{T_{s}}\left(N\! -\! n_{\tau}\right)e^{j{n_{\tau}}\phi} \!+\!
\frac{\lambda}{T_{s}}\left(N \!-\!n_{\tau} \!-\! 1\right)e^{j\left(n_{\tau}
\!+\! 1\right)\phi}.
\end{align}
It should be noted that the error probability of SAO detection depends on the selected training sequences $\mathbf{t}_{i}$. Consider the optimal training sequences in Proposition 3; we have
\begin{align}\nonumber
|\Psi_{1}\!\left(\tau\right)|\!\leq\! \frac{\lambda}{T_{s}},
|\Psi_{2}\!\left(\tau, \phi\right)|\!\leq\! N \!-\! \frac{\tau}{T_{s}}.
\end{align}
Recalling \eqref{EQ10}, we have
\begin{align}\nonumber
d_{\mathcal{H}_{0}}\! \geq\!\|\mathbf{h}\|^{2}\varsigma\left({\tau}\right)N \!+\!
2\mathfrak{R}\{h_{1}^{H}h_{2}\varsigma\left({\tau}\right)\} \!=\!
d\left(\mathbf{h}|\mathcal{H}_{0}\right)
\end{align}
where
\begin{align}
\varsigma\left({\tau}\right) \!&=\! 1 \!-\!\frac{\left(N \!-\!\frac{\tau}{T_{s}}\right)^{2}}{N^{2}\! -\!
\left(\frac{\lambda}{T_{s}}\right)^{2}}.\nonumber
\end{align}
It can be verified that $d\left(\mathbf{h}|\mathcal{H}_{0}\right)$ is positive for arbitrary non-zero $\mathbf{h}$. Similarly, we can also show that $d\left(\mathbf{h}|\mathcal{H}_{1}\right) \!>\! 0$ for all non-zero $\mathbf{h}$. For a sufficiently large $N$, the first term in the right side of $d_{\mathcal{H}_{\vartheta}}$ is much larger than the second term, and thus the second term can be neglected. This completes the proof.

\section{Proof of Lemma 2}
By observing the term $\mathbf{T}^{H}_{\mathcal{H}_{\vartheta}}\mathbf{T}_{\mathcal{H}_{\vartheta}}$
in $\mathbf{Z}_{\mathcal{H}_{\vartheta}}$, we can see that it can be treated as a diagonally dominant matrix when $N$ is sufficiently large. Thus its inverse can be approximately obtained by $\frac{1}{N}\mathbf{I}_{2}$, and $\tilde{\mathbf{n}}_{\mathcal{H}_{0}}$ can be simplified as
\begin{align}\nonumber
\tilde{\mathbf{n}}_{\mathcal{H}_{0}} \!&\approx\!
\mathbf{h}^{H}\mathbf{T}^{H}_{\mathcal{H}_{0}}\mathbf{n}_{R} \!+\!
\mathbf{n}^{H}_{R}\mathbf{T}_{\mathcal{H}_{0}}\mathbf{h} \!+\!
\frac{1}{N}\left(\mathbf{n}^{H}_{R}\mathbf{T}_{\mathcal{H}_{0}}\mathbf{T}^{H}_{\mathcal{H}_{0}}\mathbf{n}_{R}\!-\!
\mathbf{n}^{H}_{R}\mathbf{T}_{\mathcal{H}_{1}}\mathbf{T}^{H}_{\mathcal{H}_{1}}\mathbf{n}_{R}\right)\nonumber\\
\!&-\!\frac{1}{N}\left(\mathbf{h}^{H}\mathbf{T}^{H}_{\mathcal{H}_{0}}\mathbf{T}_{\mathcal{H}_{1}}\mathbf{T}^{H}_{\mathcal{H}_{1}}\mathbf{n}_{R}
\!+\!
\mathbf{n}^{H}_{R}\mathbf{T}_{\mathcal{H}_{1}}\mathbf{T}^{H}_{\mathcal{H}_{1}}\mathbf{T}_{\mathcal{H}_{0}}\mathbf{h}\right)\nonumber.
\end{align}
Note that the variance of the second term of $\tilde{\mathbf{n}}_{\mathcal{H}_{0}}$ is of order $\mathcal{O}\!\left(\left(\frac{N_{0}}{P_{s}T_{s}}\right)^{2}\right)$, while the first and third terms are both of order $\mathcal{O}\!\left(\frac{N_{0}}{P_{s}T_{s}}\right)$. In the high SNR region, we need to consider only the first and third terms. Then $\tilde{\mathbf{n}}_{\mathcal{H}_{0}}$ simplifies to
\begin{align}\nonumber
\tilde{\mathbf{n}}_{\mathcal{H}_{0}}\!\approx\!
2\mathfrak{R}\big\{\mathbf{h}^{H}\mathbf{T}^{H}_{\mathcal{H}_{0}}\mathbf{n}_{R}\big\}
\!-\!\frac{2}{NP_{S}}\mathfrak{R}\big\{\mathbf{h}^{H}\mathbf{T}^{H}_{\mathcal{H}_{0}}\mathbf{T}_{\mathcal{H}_{1}}\mathbf{T}^{H}_{\mathcal{H}_{1}}\mathbf{n}_{R}\big\}.\nonumber
\end{align}
After some manipulation, the above equation can be further written as
\begin{align}\nonumber
\tilde{\mathbf{n}}_{\mathcal{H}_{0}} \!&\approx\! 2\mathfrak{R}
\bigg\{\sum_{k =
1}^{n_{\tau}}\eta_{1}\!\left(k\right)n_{R}\!\left[k\right] \!+\!
\sqrt{\frac{\lambda}{T_{s}}}\eta_{1}\!\left(n_{\tau}\!+\!1\right)n_{R}\!\left[n_{\tau}\!+\!1\right]
\!+\! \sqrt{\frac{\lambda}{T_{s}}}\eta_{2}\!\left(2N\!
-\!n_{\tau}\!+\!1\right)n_{R}\!\left[2N \!-\! n_{\tau}\!+\!1\right] \nonumber\\\!&+\!\!\!\!\!\!\sum_{k = 2N\!-\!n_{\tau} \!+\! 2}^{2N \!+\!1}\!\!\!\!\!\!\!\eta_{2}\!\left(k\right)n_{R}\!\left[k\right]
\!+\!\!\!\!\!\!\!\!\sum_{i\!=\!1,k\!=\!n_{\tau}\!+\!2i}^{i\!=\!N\!-\!n_{\tau}}\!\!\sqrt{\frac{\lambda}{T_{s}}}\left[\eta_{1}\!\left(k\right)
\!+\!\eta_{2}\!\left(k\right)\right]n_{R}\!\left[k\right] \!+\!\!\!\!\!\!\!\!\sum_{i\!=\!1,k\!=\!n_{\tau}\!+\!2i\!+\!1}^{i\!=\!N\!-\!n_{\tau}\!-\!1}\sqrt{\frac{T_{s}\! -\!\lambda}{T_{s}}}\left[\eta_{1}\!\left(k\right)\!+\!\eta_{2}\!\left(k\right)\right]n_{R}\!\left[k\right]\bigg\}\nonumber
\end{align}
with
\begin{align}\nonumber
\eta_{1}\!\left(k\right) \!=\!e^{j\phi_{1}\left(k \!-\! 1\right)} h_{1}^{H}\!-\!\frac{\Psi_{2}\!\left(\tau, \phi_{2}\right)}{N}h_{2}^{H},
\eta_{2}\!\left(k\right)\! =\!e^{j\phi_{2}\left(k \!-\! n_{\tau} \!-\!1\right)}h_{2}^{H}\!-\! \frac{\Psi_{2}\!\left(\tau, \phi_{1}\right)}{N}h_{1}^{H}.
\end{align}
With the optimal training sequences, we have
\begin{align}\nonumber
|\eta_{1}\!\left(k\right)|^{2}\!&\leq\! \left(|h_{1}|^{2} \!+\!
\left(1 \!-\! \frac{\tau}{NT_{s}}\right)^{2}|h_{2}|^{2}\right),\\
|\eta_{2}\!\left(k\right)|^{2}\!&\leq\! \left(|h_{2}|^{2} \!+\! \left(1 \!-\!
\frac{\tau}{NT_{s}}\right)^{2}|h_{1}|^{2}\right),\nonumber
\end{align}
\begin{align}\nonumber
|\eta_{1}\!\left(k\right)\!+\!\eta_{2}\!\left(k\right)|^{2}\!\leq\!
\left(\frac{\tau}{NT_{s}}\right)^{2}\left(|h_{1}|^{2} \!+\!
|h_{2}|^{2}\right).\nonumber
\end{align}
Clearly, it can be observed that $\tilde{n}_{E}$ has the real Gaussian distribution with zero mean, and we have
\begin{align}\nonumber
\upsilon_{\mathcal{H}_{0}}\! \leq\! 2\left(2 \!-\!\frac{\tau}{NT_{s}}\right)\frac{\tau}{T_{s}}\frac{N_{0}}{P_{s}T_{s}}\left(|h_{1}|^{2}
\!+\! |h_{2}|^{2}\right),
\end{align}
and similarly for $\upsilon_{\mathcal{H}_{1}}$. This completes the proof.

\begin{IEEEbiography}[{\includegraphics[width=1in,height=1.25in]{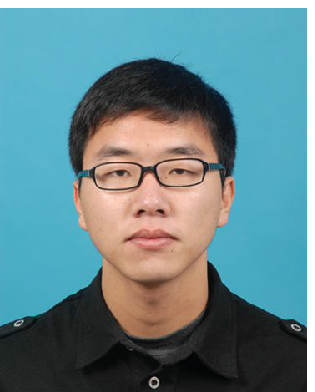}}]{Xinqian Xie}
received the B.S. degree in telecommunication engineering from Beijing University of Posts
and Communications (BUPT), China, in 2010. He is currently pursuing the Ph.D. degree in Wireless
Signal Processing and Networks Lab at BUPT. His research interests include cooperative communications, estimation and detection theory.
\end{IEEEbiography}

\begin{IEEEbiography}[{\includegraphics[width=1in,height=1.25in]{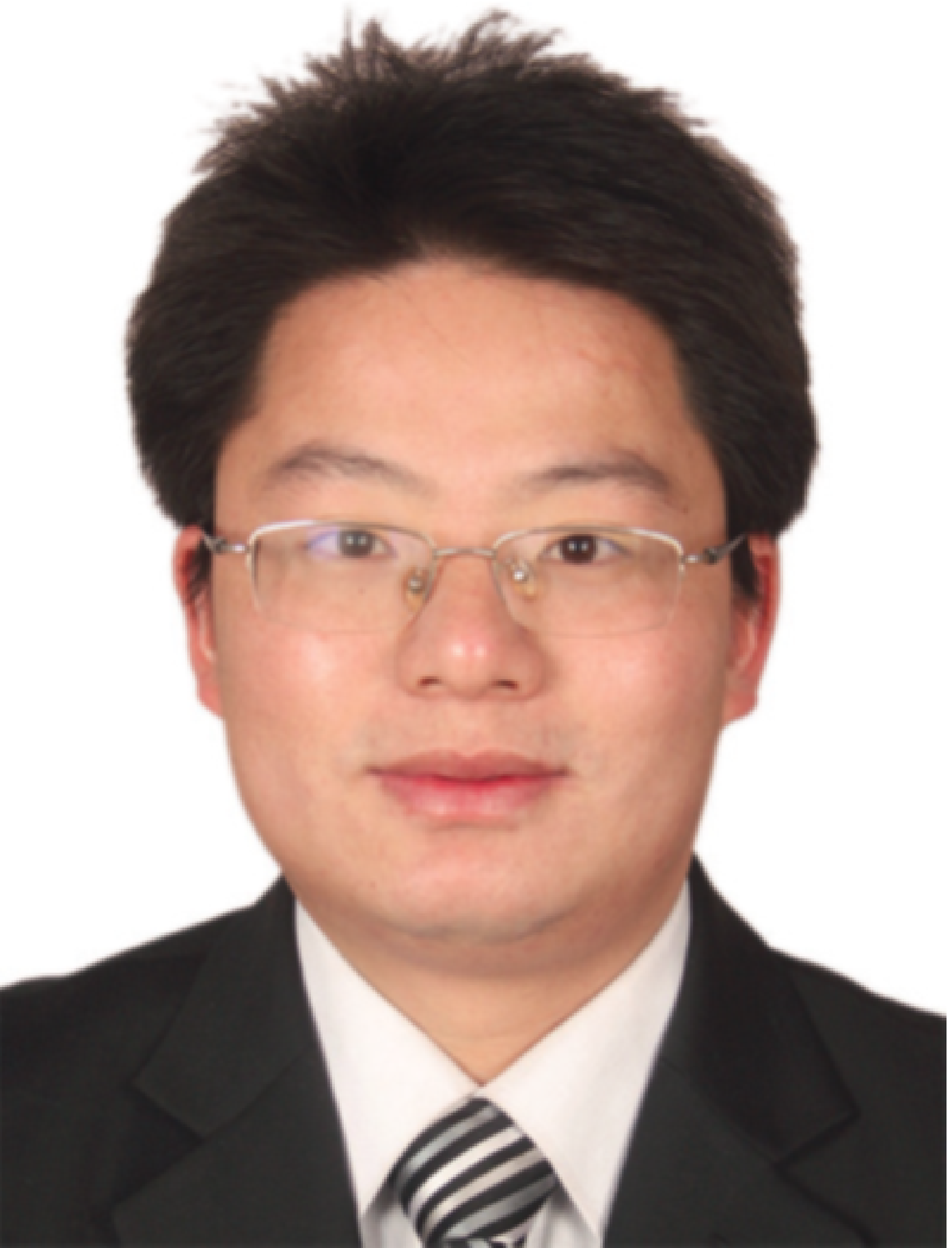}}]{Mugen Peng}
(M'05--SM'11) received the B.E. degree in Electronics Engineering
from Nanjing University of Posts \& Telecommunications, China in
2000 and a PhD degree in Communication and Information System from
the Beijing University of Posts \& Telecommunications (BUPT), China
in 2005. After the PhD graduation, he joined in BUPT, and has become
a full professor with the school of information and communication
engineering in BUPT since Oct. 2012. During 2014, he is also an
academic visiting fellow in Princeton University, USA. He is leading
a research group focusing on wireless transmission and networking
technologies in the Key Laboratory of Universal Wireless
Communications (Ministry of Education) at BUPT, China. His main
research areas include wireless communication theory, radio signal
processing and convex optimizations, with particular interests in
cooperative communication, radio network coding, self-organization
networking, heterogeneous networking, and cloud communication. He
has authored/coauthored over 40 refereed IEEE journal papers and
over 200 conference proceeding papers.

Dr. Peng is currently on the Editorial/Associate Editorial Board of
IEEE Access, International Journal of Antennas and Propagation
(IJAP), China Communication, and International Journal of
Communications System (IJCS). He has been the guest leading editor
for the special issues in IEEE Wireless Communications, IJAP and the
International Journal of Distributed Sensor Net- works (IJDSN). He
was the guest editor of IET Communications. He is serving as the
track chair for GameNets 2014, WCSP 2013 and ICCT 2011, the workshop
co-chair of ChinaCom 2012, and the leading co-chair for So-HetNets
in IEEE WCNC 2014, SON-HetNet 2013 in IEEE PIMRC 2013, and SON 2013
in ChinaCom 2012. Dr. Peng was honored with the Best Paper Award in
CIT 2014, ICCTA 2011, IC-BNMT 2010, and IET CCWMC 2009. He was
awarded the first Grade Award of Technological Invention Award in
Ministry of Education of China for his excellent research work on
the hierarchical cooperative communication theory and technologies,
and the Second Grade Award of Scientific \& Technical Progress from
China Institute of Communications for his excellent research work on
the co-existence of multi-radio access networks and the 3G spectrum
management in China.
\end{IEEEbiography}

\begin{IEEEbiography}[{\includegraphics[width=1in,height=1.25in]{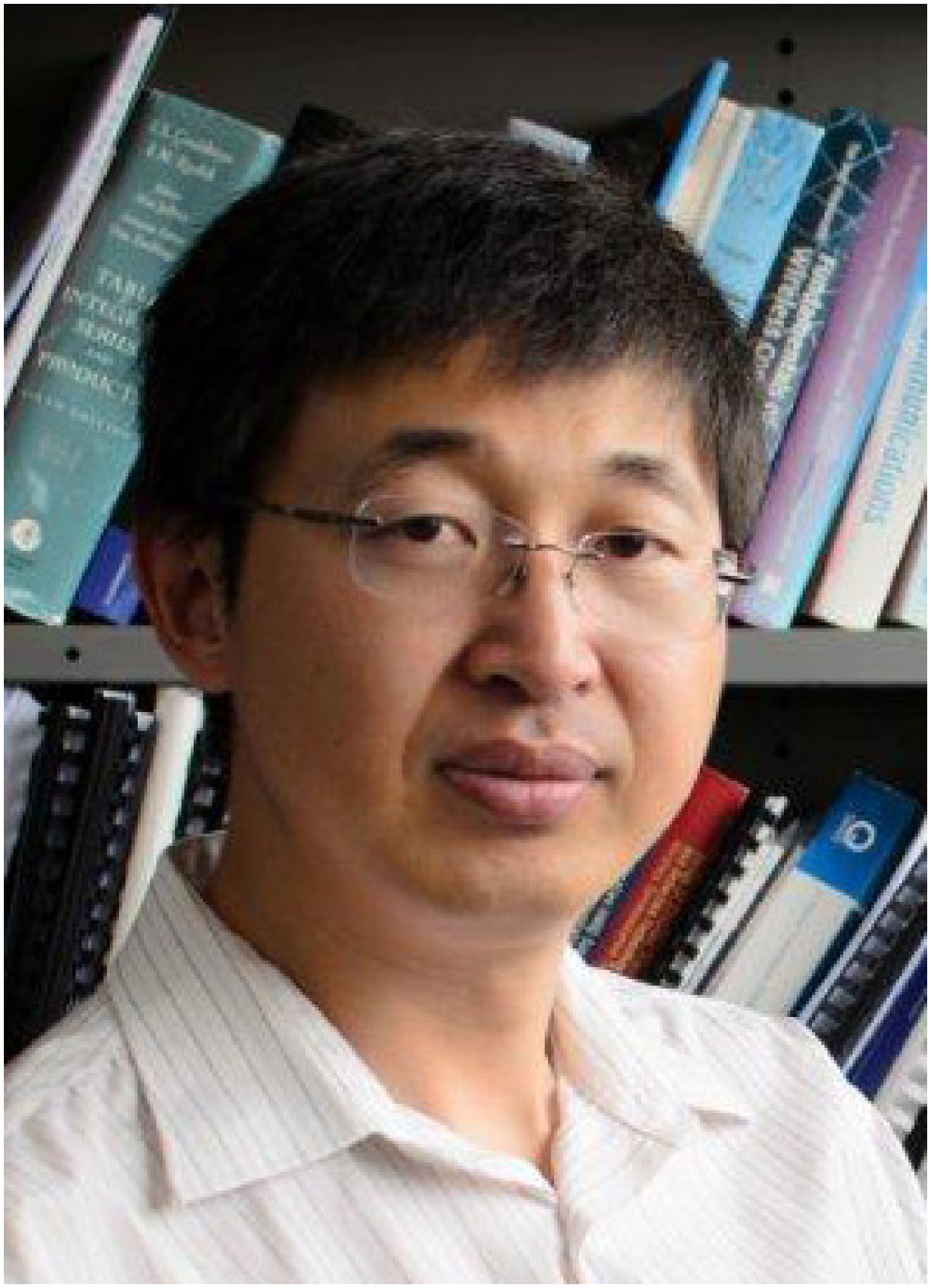}}]{Yonghui Li}
(M'04-SM'09) received his PhD degree in November 2002 from Beijing University of Aeronautics and Astronautics. From 1999-2003, he was affiliated with Linkair Communication Inc, where he held a position of project manager with responsibility for the design of physical layer solutions for the LAS-CDMA system. Since 2003, he has been with the Centre of Excellence in Telecommunications, the University of Sydney, Australia. He is now an Associate Professor in School of Electrical and Information Engineering, University of Sydney. He was the Australian Queen Elizabeth II Fellow and is currently the Australian Future Fellow.
His current research interests are in the area of wireless communications, with a particular focus on MIMO, cooperative communications, coding techniques and wireless sensor networks. He holds a number of patents granted and pending in these fields. He is an executive editor for European Transactions on Telecommunications (ETT). He has also been involved in the technical committee of several international conferences, such as ICC, Globecom, etc.
\end{IEEEbiography}

\begin{IEEEbiography}[{\includegraphics[width=1in,height=1.25in]{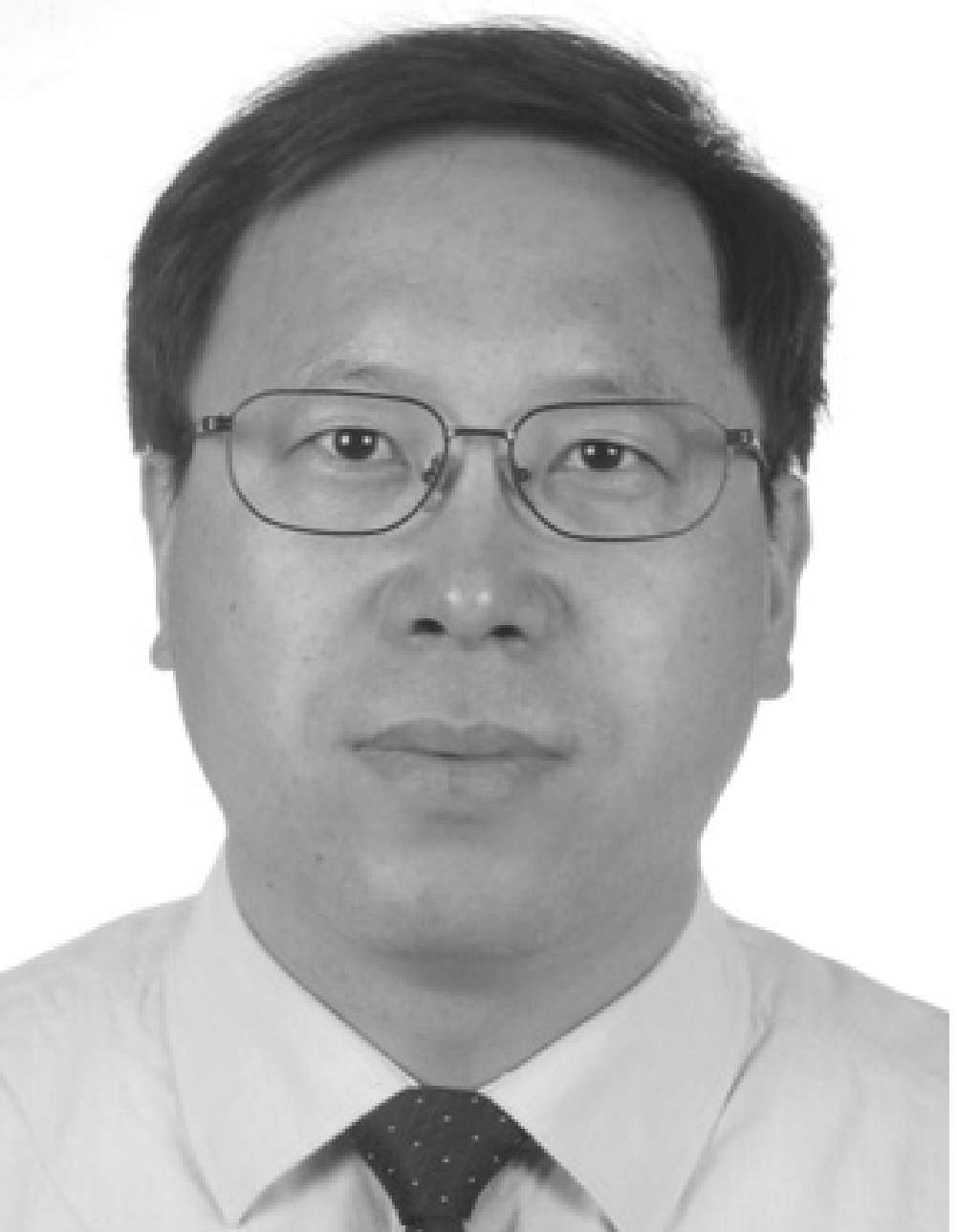}}]{Wenbo Wang}
is currently the dean of Telecommunication Engineering at Beijing University of Posts and Telecommunications (BUPT). He received the BS degree, the MS and Ph.D. Degrees from BUPT in 1986, 1989 and 1992 respectively. Now he is the Assistant Director of academic committee of Key Laboratory of Universal Wireless Communication (Ministry of Education) in BUPT. His research interests include radio transmission technology, Wireless network theory, Broadband wireless access and Software radio technology. Prof. Wenbo Wang has published more than 200 journal and international conference papers and holds 12 patents and has published six books.
\end{IEEEbiography}

\begin{IEEEbiography}[{\includegraphics[width=1in,height=1.25in]{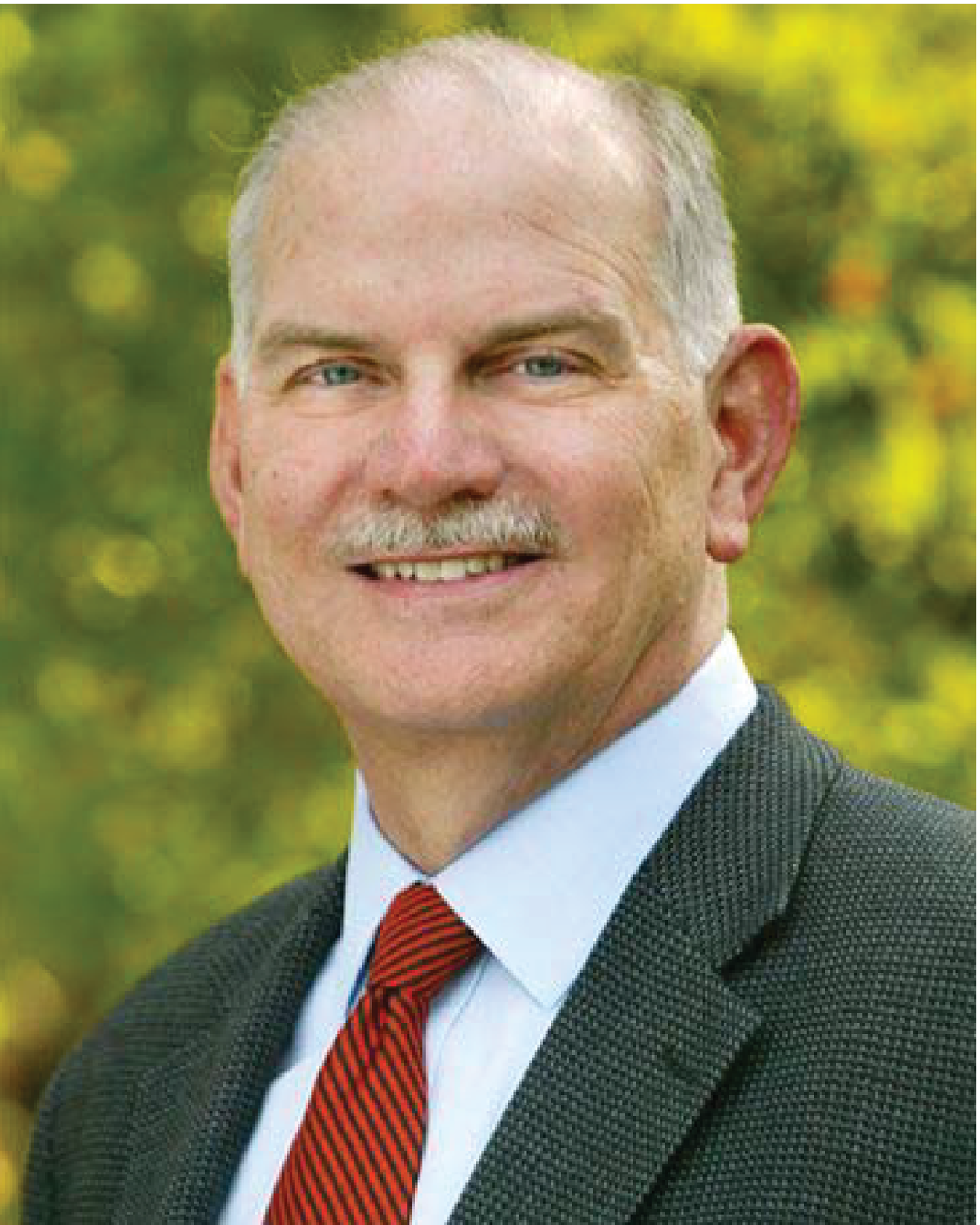}}]{H. Vincent Poor}
(S'72, M'77, SM'82, F'87) received the Ph.D. degree in EECS from Princeton University in 1977.  From 1977 until 1990, he was on the faculty of the University of Illinois at Urbana-Champaign. Since 1990 he has been on the faculty at Princeton, where he is the Michael Henry Strater University Professor of Electrical Engineering and Dean of the School of Engineering and Applied Science. Dr. Poor's research interests are in the areas of stochastic analysis, statistical signal processing, and information theory, and their applications in wireless networks and related fields such as social networks and smart grid. Among his publications in these areas are the recent books \textit{Principles of Cognitive Radio} (Cambridge, 2013) and \textit{Mechanisms and Games for Dynamic Spectrum Allocation} (Cambridge University Press, 2014).

Dr. Poor is a member of the National Academy of Engineering and the
National Academy of Sciences, and a foreign member of Academia
Europaea and the Royal Society. He is also a fellow of the American
Academy of Arts and Sciences, the Royal Academy of Engineering
(U.K.), and the Royal Society of Edinburgh. He received the
Technical Achievement and Society Awards of the IEEE Signal
Processing Society in 2007 and 2011, respectively. Recent
recognition of his work includes the 2014 URSI Booker Gold Medal,
and honorary doctorates from several universities in Europe and
Asia, including an honorary D.Sc. from Aalto University in 2014.
\end{IEEEbiography}
\end{document}